\documentclass[journal]{IEEEtran}
\ifCLASSINFOpdf
\else
\fi
\usepackage{algorithm}
\usepackage{algpseudocode}
\usepackage{amsmath, amssymb}
\usepackage{mathtools}
\usepackage{algorithm}
\usepackage{algpseudocode}
\usepackage{amsmath,amssymb,mathtools}
\usepackage{xcolor}
\usepackage{bm}
\usepackage{amsmath, amsthm}

\usepackage{booktabs}
\usepackage{amsthm}
\newtheorem{lemma}{Lemma}
\usepackage{siunitx}
\usepackage{amsthm}
\usepackage{caption}
\newtheorem{proposition}{Proposition}
\usepackage{graphicx}
\usepackage{placeins}
\usepackage{subcaption}
\begin{document}
%
\title{CRB-Guided Sensing and Resource Allocation for Human Pose Prediction in Integrated Sensing, Communication, and Computation Systems}
%
%
%

\author{Zhonghao Liu, Yahao Ding, Jiaxiang Wang, Zhaohui Yang, Abdol Hamid Aghvami,~\IEEEmembership{Fellow,~IEEE},
        and Mohammad Shikh-Bahaei,~\IEEEmembership{Senior Member,~IEEE}
\thanks{Zhonghao Liu, Yahao Ding, Jiaxiang Wang, Abdol Hamid Aghvami, and Mohammad Shikh-Bahaei are with the Department of Engineering, King's College London, London, UK (emails: zhonghao.liu@kcl.ac.uk; yahao.ding@kcl.ac.uk; jiaxiang.wang@kcl.ac.uk; hamid.aghvami@kcl.ac.uk; m.sbahaei@kcl.ac.uk).}
\thanks{Zhaohui Yang is with the College of Information Science and Electronic Engineering, Zhejiang University, Hangzhou, Zhejiang 310027, China, and Zhejiang Provincial Key Lab of Information Processing, Communication and Networking (IPCAN), Hangzhou, Zhejiang, 310007, China (email: yang\_zhaohui@zju.edu.cn).}
}



\maketitle

\begin{abstract}
Integrated sensing, communication, and computation (ISCC) provides a promising framework for indoor human-centric applications. In these applications, short-term human pose prediction facilitates continuous human pose tracking and proactive resource allocation. This paper proposes a Cramér–Rao bound (CRB)-guided sensing framework and investigates a problem of minimizing prediction error in resource-constrained ISCC systems. Specifically, a pose prediction model (ET-Mamba) is first developed to predict human joint positions for continuous tracking. To account for computation-resource limitations, lightweight prediction heads are attached to different inference layers, enabling adaptive-depth pose prediction. A CRB-guided perturbation strategy is then introduced to translate sensing uncertainty at different sensing SNR levels into point-cloud perturbations. Based on that, an empirical relationship among pose prediction error, sensing SNR, and model inference depth is established. Furthermore, to improve prediction accuracy under limited resources, this paper formulates a resource allocation optimization problem that minimizes the pose prediction error by jointly optimizing the beamforming matrix, model inference depth, and computation frequency. To solve this mixed-integer non-convex optimization problem, we propose an alternating optimization (AO)-based algorithm, where closed-form updates and semidefinite programming (SDP) are integrated into the iterative solution process. Simulation results show that the proposed method effectively improves pose prediction performance by up to 35\% under resource constraints, verifying the effectiveness of conducting joint sensing, communication, and computation design in ISCC systems.
\end{abstract}


%
\IEEEpeerreviewmaketitle

\section{Introduction}
Integrated sensing, communication, and computation (ISCC) is regarded as a promising technological framework for supporting intelligent human-centric applications in next-generation wireless networks \cite{10812728, chen2025sensing}. By integrating sensing, communication, and computation into a unified platform~\cite{9737357}, ISCC systems can support real-time indoor applications such as human pose tracking, gesture recognition, and activity monitoring~\cite{chen2025sensing, liu2025integrated}. Short-term human pose prediction plays a key role in indoor human-centric sensing, as it enables continuous pose tracking and provides temporal motion cues for high-level semantic understanding of human activities \cite{wu2024mmhpe}.

For these sensing tasks, mmWave sensing, owing to its fine spatial resolution and non-contact sensing capability, provides an effective modality for human-centric perception \cite{wu2024mmhpe,engel2025advanced}. In mmWave-based human sensing, reflected echoes are typically processed into sparse point clouds, which encode spatial and motion-related information of human body movements. By exploiting temporal dependencies across historical point-cloud observations, future human joint positions can be inferred to support short-term pose prediction.
Accurate short-term pose prediction relies on effective temporal dependency modeling over historical point-cloud sequences.


Existing studies related to this work can be broadly categorized into two research lines: physical-layer resource allocation and AI-driven mmWave point-cloud-based human sensing. The first line focuses on physical-layer optimization in ISAC/ISCC systems, spanning multiple directions including resource allocation \cite{10008661, 10556582, yao2025hybrid, yao2026energy}, interference management \cite{liu2023snr, 10626219}, signal design \cite{li2024mimo}, and joint transceiver optimization \cite{11030617, 10734321}. These works typically aim to improve communication quality of service (QoS)~\cite{zhang2025communication} or optimize sensing metrics, such as sensing SINR, resolution, and the Cramér–Rao bound (CRB)~\cite{10008661,10251151,10217169}. Recent studies have also investigated learning-aware resource allocation in wireless networks~\cite{wang2026semantic, wang2025generative, liu2026crb, zhang2026tt}. Specifically, the work~\cite{liu2023snr} took the residual self-interference under imperfect cancellation into consideration and maximized the communication rate subject to the CRB and SNR constraints. The work~\cite{10251151} proposed an algorithm for multi-target sensing with known or unknown target prior information, which minimizes the CRB while satisfying the communication QoS. The work~\cite{10217169} proposed an algorithm to optimize the BS transmit covariance matrix, trading off sensing CRB against the communication rate. 

Another line of research focuses on AI-driven mmWave point-cloud-based human sensing. Compared with physical-layer optimization studies, human pose prediction studies usually evaluate sensing performance using metrics such as Mean Per Joint Position Error (MPJPE) and Procrustes-Aligned Mean Per Joint Position Error (PA-MPJPE), which directly measure the accuracy of predicted joint positions. Existing studies mainly concentrate on real-time human pose estimation, aiming to obtain joint positions from sparse and noisy radar point clouds~\cite{wu2024mmhpe,tang2025gf,11036600}. Specifically,~\cite{wu2024mmhpe} proposed mmHPE, which generates stable and accurate point clouds by enhancing target boundary detection, and performs human pose estimation in a multi-scale manner. In~\cite{tang2025gf}, the authors proposed GF-DecNet, which integrates set abstraction and geometry-aware feature decoupling to extract local geometric and global contextual features from radar point clouds. The work~\cite{11036600} proposed mmPoint-Attention, which achieves accurate 3D joint estimation and action classification by contextual information and attention mechanisms. Attention-based architectures have shown strong representation capability, but their quadratic complexity with respect to the sequence length may introduce high computational overhead and inference latency. Mamba-based state-space models provide a computationally efficient alternative by enabling linear-complexity sequence modeling through selective state updates. 
{Recent studies have explored Mamba architectures for point-cloud understanding from different perspectives~\cite{liang2024pointmamba, diao2025zigzagpointmamba, liu2025mamba4d, chen2025stpm}. Specifically,~\cite{liang2024pointmamba} first introduced Mamba into point cloud analysis and proposed the PointMamba architecture. The work~\cite{diao2025zigzagpointmamba} improved point cloud serialization through a spatially continuous zigzag scanning strategy. In~\cite{liu2025mamba4d}, the authors extended Mamba to dynamic point cloud video understanding by modeling spatial and temporal features. The work~\cite{chen2025stpm} proposed a spatial-temporal bidirectional Mamba architecture to jointly model spatial structures and temporal dependencies.}

Although these studies on physical-layer optimization provide effective schemes, they do not explicitly consider the downstream sensing error induced by resource allocation. Meanwhile, although these learning-based mmWave sensing methods have significantly improved performance through model architecture design, they do not consider the impact of limited resources on downstream pose prediction error. In practical resource-limited ISCC systems, pose prediction performance is jointly affected by sensing and computation resources. Resource allocation decisions may induce sensing SNR fluctuations, leading to stochastic point-cloud jittering before model inference~\cite{quang2026diffusion}. Meanwhile, limited computation resources may prevent full-depth inference within each operating slot. These factors jointly degrade pose prediction accuracy and may further affect future tracking performance.

This indicates a clear research gap: existing studies still lack a quantitative analysis that links point-cloud perturbation, computation-limited inference, and the performance of the sensing task in mmWave ISCC systems.
To address the above issues, this paper proposes a CRB-guided framework design and resource allocation scheme for mmWave human pose prediction in ISCC indoor systems. The central objective of this work is to minimize the total pose prediction error through joint optimization of sensing, communication, and computation resource allocation. The main contributions of this paper are summarized as follows:
\begin{itemize}
   \item We propose a CRB-guided mmWave human pose prediction and tracking framework for resource-constrained ISCC systems. In this framework, an ET-Mamba-based prediction model is developed to predict future human joint positions from historical mmWave point-cloud sequences, where lightweight prediction heads are attached after different inference layers to support adaptive model-depth execution under computation constraints.
   \item We formulate a joint resource allocation problem for minimizing the total pose prediction error to support continuous human pose prediction and tracking. To the best of our knowledge, this is the first work that jointly considers point-cloud perturbation, model inference depth, and resource allocation in ISCC systems.
  \item To solve this problem, we first develop a sigmoid-variant empirical model to characterize the relationship among sensing SNR, model inference depth, and pose prediction error. Based on this model, we propose an AO-based algorithm, where closed-form solutions and SDP are integrated to solve the formulated problem efficiently. 
  \item Simulation results indicate that the proposed ET-Mamba achieves an MPJPE of 3.24~cm and a PA-MPJPE of 2.20~cm, outperforming representative point-cloud learning baselines. The proposed CRB-guided framework enhances prediction robustness, reducing the MPJPE by approximately 73.9\% in the low-SNR regime. Moreover, the proposed optimization method reduces the MPJPE by up to approximately 35\% compared with the baseline.
\end{itemize}

The rest of the paper is organized as follows. The system model is described in Section~II. The CRB-guided pose prediction framework is proposed in Section~III. The optimization problem is formulated in Section IV, and solution process is presented in Section V. The results and analysis of the proposed method are described in Section~VI. Finally, the
conclusions are drawn in Section VII. {Throughout this paper, $\mathrm{Tr}(\cdot)$, $\|\cdot\|_2$, and $\odot$ denote the trace of a matrix, the $\ell_2$ norm, and the Hadamard product, respectively. The main notations used throughout this paper are summarized in Table~\ref{tab:notations}.}

\begin{figure}[!t]
    \centering   \vspace{-1.3em}\includegraphics[width=0.45\textwidth]{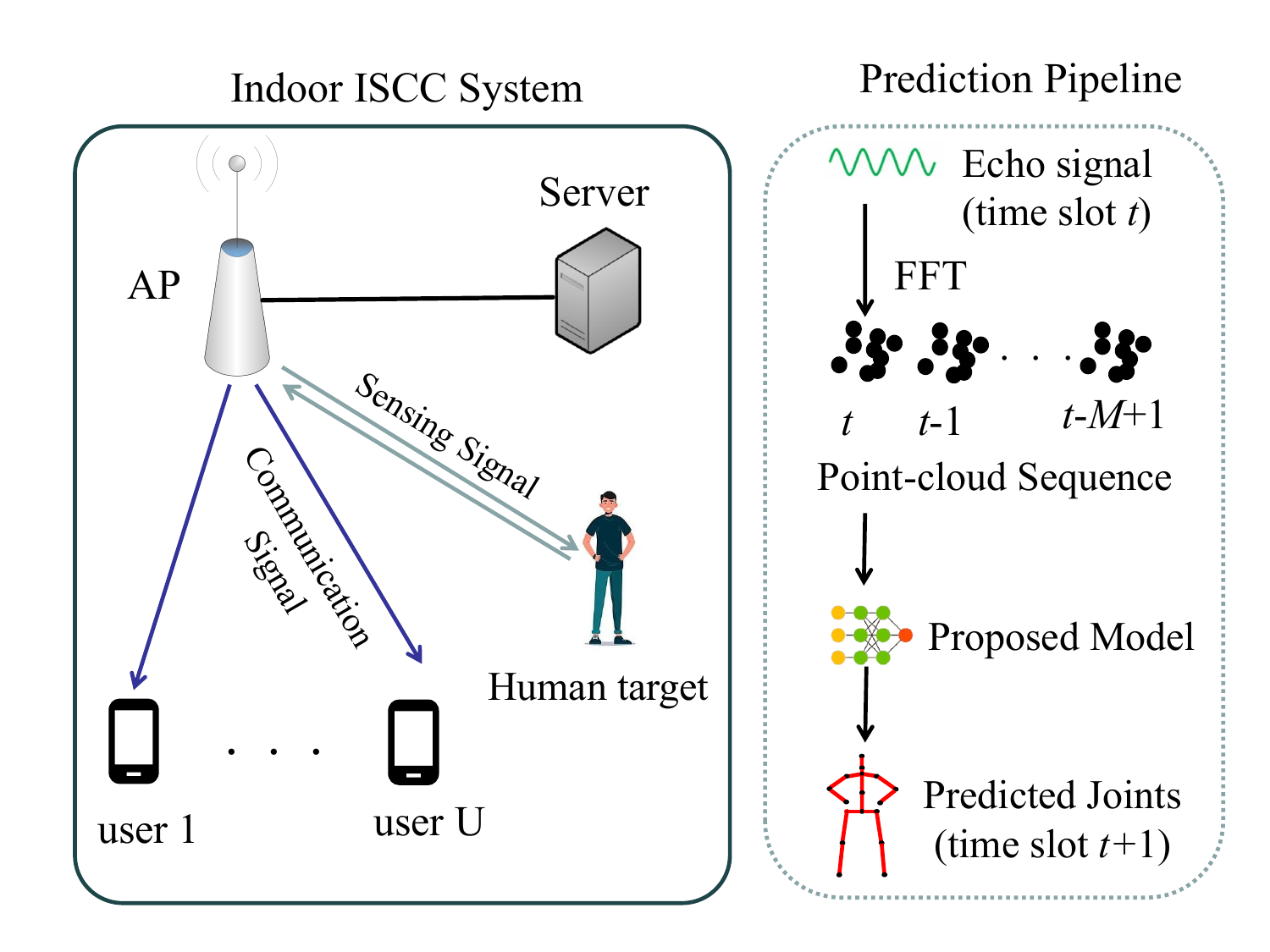}  
    \caption{Illustration of the considered indoor mmWave ISCC system for human pose prediction and tracking.}
    \vspace{-1.8em}
    \label{fig:1}
\end{figure}

\section{System Model}

As shown in Fig.~\ref{fig:1}, a mmWave access point (AP) equipped with a uniform planar array (UPA) is deployed in the indoor ISCC system to support three capabilities: high-resolution sensing for the human target, communication with $U$ users, and local computing for point-cloud generation and pose prediction. At the beginning of each time slot, the system adaptively adjusts the resources according to the predicted human joint positions to ensure effective coverage of all human joints and maintain the human body within the sensing region. {The system serves communication users $u\in\mathcal{U}$ and operates over time slots indexed by $t\in\mathcal{N}$.} 

Specifically, as shown in Fig.~\ref{fig:time}, the operational workflow of the considered ISCC system within one time slot consists of the following four stages:
\begin{itemize}
    \item[1)] Sensing: During the sensing duration of time slot $t$, the AP transmits sensing signals toward the human target based on the predicted joint positions at time slot $t-1$. 
    \item[2)] Point-cloud generation: {During the sensing process, the AP continuously receives echoes reflected from the human body, while the server processes the received sensing signals online to generate mmWave point clouds.}
    \item[3)] Inference: {After point-cloud generation is completed, the server uses the generated point clouds together with historical point clouds to predict joint positions for time slot $t+1$ of the sensing target, providing prior information for resource allocation in the next time slot.} 
    \item[4)] Communication: Throughout time slot $t$, the AP maintains data transmission with multiple users.
\end{itemize}
\begin{figure}[!t]
    \centering    \includegraphics[width=0.4\textwidth]{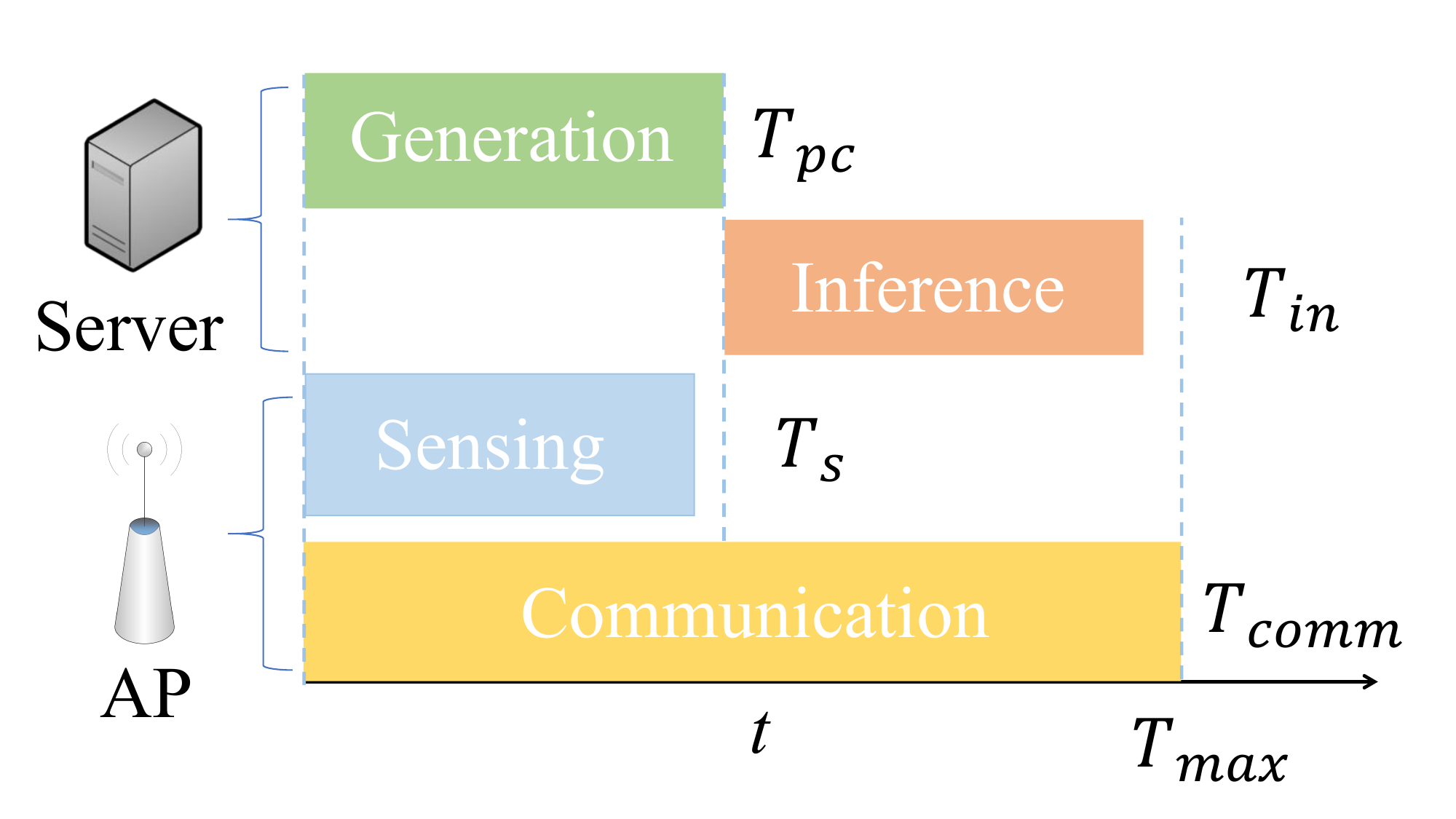}  
    \caption{Time slot structure of the considered ISCC system}
    \label{fig:time}
    \vspace{-1.5em}
\end{figure}

Without loss of generality, the UPA is assumed to be placed on the \(Oxz\) plane and consists of \(N_t=N_xN_z\) antennas. To eliminate mutual interference between communication and sensing, the system achieves decoupling through a frequency division multiplexing (FDM) strategy, whereby a dedicated frequency band is allocated for sensing tasks. At each time slot \( t \), the AP transmits a signal that consists of both communication and sensing components, which can be written as  
\begin{equation}
\mathbf{x}[t] =  \mathbf{w}_c \mathbf{s}_c[t] +  \mathbf{w}_r \mathbf{s}_r[t],
\end{equation}
{where \( \mathbf{s}_c[t] \in \mathbb{C}^{U \times 1} \) denotes the communication symbol vector for \( U \) users at time slot \( t \),} and \( \mathbf{w}_c = [\mathbf{w}_{c,1}, \dots, \mathbf{w}_{c,U}]\in \mathbb{C}^{N_t \times U} \) is the communication beamforming matrix. The sensing waveform \( \mathbf{s}_r[t] \in \mathbb{C} \) is a scalar symbol transmitted for target probing, and \( \mathbf{w}_r \in \mathbb{C}^{N_t} \) denotes the sensing beamforming vector for the sensing target. The communication symbols satisfy the orthogonality and normalization condition \( \mathbb{E}[\mathbf{s}_c[t] \mathbf{s}_c^H[t]] = \mathbf{I}_U \), while the sensing waveform is power-normalized as \( \mathbb{E}[|\mathbf{s}_r[t]|^2] = 1 \). $\mathbf{s}_c[t]$ and $\mathbf{s}_r[t]$ are statistically independent, satisfying \( \mathbb{E}[\mathbf{s}_c[t] \mathbf{s}_r^H[t]] = \mathbf{0} \). 

\begin{table}[!t]
\centering
\caption{Main notations}
\label{tab:notations}
\footnotesize
\renewcommand{\arraystretch}{1.05}

\begin{tabular}{@{}p{0.18\columnwidth} p{0.76\columnwidth}@{}}
\toprule
\textbf{Symbol} & \textbf{Description} \\
\midrule

{$\mathcal{U}$}
& {Set of communication users} \\

{$U$}
& {Number of communication users} \\

{$t$}
& {Time-slot index} \\

{$\mathcal{N}$}
& {Set of time-slot indices} \\

{$N_t$}
& {Number of antennas at the AP} \\

{$N_x,N_z$}
& {Numbers of UPA elements along the $x$- and $z$-axes} \\

{$B_c$}
& {Communication bandwidth} \\

{$B_r$}
& {Sensing bandwidth} \\

{$R_u$}
& {Achievable communication rate of user $u$} \\

{$T$}
& {Duration of one time slot} \\

{$T_s$}
& {Sensing duration in one time slot} \\

{$C$}
& {Selected model inference depth} \\

{$\mathcal{C}$}
& {Set of candidate model inference depths} \\

{$C_{\max}$}
& {Maximum model inference depth} \\

{$f_{\max}$}
& {Maximum computation frequency} \\

{$\mathcal{G}$}
& {Set of final group tokens after GLF} \\

{$G$}
& {Number of local point-cloud groups} \\

{$\mathcal{J}$}
& {Set of human joints} \\

{$J$}
& {Number of human joints} \\

{$M$}
& {Number of consecutive point-cloud frames} \\

{$\mathbf{X}_i$}
& {Point-cloud data of the $i$-th frame} \\

{$\mathbf{Z}_i$}
& {Ordered token sequence of the $i$-th frame} \\

{$\mathbf{S}$}
& {Cross-temporal serialized token sequence} \\

{$N_c$}
& {Number of sensing snapshots} \\

{$N_s$}
& {FFT size for point-cloud generation} \\

{$\tau_{\mathrm{pc}}$}
& {Point-cloud generation latency} \\

{$\tau_{\mathrm{in}}$}
& {Model inference latency} \\

{$\gamma$}
& {Effective switched-capacitance coefficient} \\

{$E_s$}
& {Sensing transmission energy} \\

{$E_c$}
& {Communication transmission energy} \\

{$E_{\mathrm{pc}}$}
& {Energy consumption for point-cloud generation} \\

{$E_{\mathrm{in}}$}
& {Energy consumption for model inference} \\

{$P_{\max}$}
& {Maximum allowable average power} \\

{$m(\mathrm{SNR}_{s,o},C)$}
& {Empirical MPJPE function} \\

{$\mathbf{w}_{c,u}$}
& {Communication beamforming vector for user $u$} \\

{$\mathbf{w}_{r}$}
& {Sensing beamforming vector} \\

{$f_{\mathrm{pc}}$}
& {Computation frequency for point-cloud generation} \\

{$f_{\mathrm{in}}$}
& {Computation frequency for model inference} \\

\bottomrule
\end{tabular}
\end{table}
\subsection{Communication model}
We assume a standard half-wavelength element spacing, satisfying $d=\lambda/2$, where \(d\) and \(\lambda\) denote the element spacing and carrier wavelength, respectively. {Given the azimuth angle $\theta_u$ and elevation angle $\phi_u$ of user $u$, where $\theta_u$ is measured in the $xy$-plane from the $+y$-axis toward the $+x$-axis and $\phi_u$ is measured from the $xy$-plane toward the $+z$-axis}, the array steering vector of the UPA is defined as $\mathbf{a}_t(\theta_{u},\phi_{u}) = \mathbf{a}_x(\theta_{u},\phi_{u}) \otimes \mathbf{a}_z(\theta_{u},\phi_{u})$. {Here, $\mathbf{a}_x(\theta_{u},\phi_{u}) = [1, e^{j \frac{2\pi d}{\lambda}\sin(\theta_{u})\cos(\phi_{u})}, \ldots, e^{j (N_x-1)\frac{2\pi d}{\lambda}\sin(\theta_{u})\cos(\phi_{u})}]^T$, and $\mathbf{a}_z(\theta_{u},\phi_{u}) = [1, e^{j\frac{2\pi d}{\lambda}\sin(\phi_{u})}, \ldots, e^{j(N_z-1)\frac{2\pi d}{\lambda}\sin(\phi_{u})}]^T$.}  

Accordingly, the received signal of user $u$ at time slot $t$ can be expressed as:
\begin{equation}
y_u[t] = \mathbf{h}_{u}^H \mathbf{w}_{c,u} s_{c,u}[t] + \sum_{i \neq u} \mathbf{h}_{u}^H \mathbf{w}_{c,i} s_{c,i}[t] + n_{u}, 
\end{equation}
where $\mathbf{h}_{u}$ is the downlink channel between the AP and the $u$-th user and $n_u \sim \mathcal{CN}(0, \sigma_n^2 )$ is the Gaussian noise.
{To account for both line-of-sight (LoS) and non-line-of-sight (NLoS) propagation, the channel is modeled as $\mathbf{h}_{u}=\sqrt{\beta_u}\widetilde{\mathbf{h}}_{u}$, where $\beta_u=\beta_0(d_u/d_0)^{-\alpha}$ denotes the large-scale channel power gain, $\beta_0$ denotes the channel power gain at the reference distance $d_0=1$~m. $d_u$ is the distance between the AP and user $u$, and $\alpha$ is the path-loss exponent. The small-scale channel $\widetilde{\mathbf{h}}_{u}$ follows a Rician fading model and is expressed as $\widetilde{\mathbf{h}}_{u}=\sqrt{\frac{K_u}{K_u+1}}\mathbf{a}_{t}(\theta_u,\phi_u)+\sqrt{\frac{1}{K_u+1}}\mathbf{g}_{u}$, where $K_u$ is the Rician factor given in~\cite{han2016efficient}, $\mathbf{a}_{t}(\theta_u,\phi_u)$ represents the LoS component, and $\mathbf{g}_{u}\sim\mathcal{CN}(\mathbf{0},\mathbf{I}_{N_t})$ represents the NLoS scattering component.}

The signal-to-interference-plus-noise ratio (SINR) of user $u$ can be expressed as:
\begin{equation}
\text{SINR}_{u} = \frac{|\mathbf{h}_{u}^H \mathbf{w}_{c,u}|^2}{\sum_{i \neq u} |\mathbf{h}_{u}^H \mathbf{w}_{c,i}|^2 +\sigma_n^2}.
\end{equation}
The achievable communication rate of user $u$ is then given by $R_{u} = B_c \log_2(1 + \text{SINR}_{u})$.
\subsection{Sensing model}
At time slot $t$, the received sensing signal at the AP can be expressed as
\begin{equation}
\mathbf y_s[t] =\mathbf e[t] + \mathbf c[t]+\mathbf n_{s},
\label{eq:rx_signal}
\end{equation}
where $\mathbf e[t]$ is the echo signal at time slot $t$, $\mathbf c[t]$ denotes the static clutter component caused by environmental reflectors, and $\mathbf n_{s}$ is the additive sensing noise, modeled as $\mathbf n_s \sim \mathcal{CN}(\mathbf 0, \sigma_z^2\mathbf I_{N_t} )$.
We assume $\bm{a}_t(\theta,\phi)\in\mathbb{C}^{N_t}$ and $\bm{a}_r(\theta,\phi)\in\mathbb{C}^{N_t}$ are the transmit and receive steering vectors, respectively.
\begin{figure*}[!t]
  \centering  \includegraphics[width=0.98\textwidth]{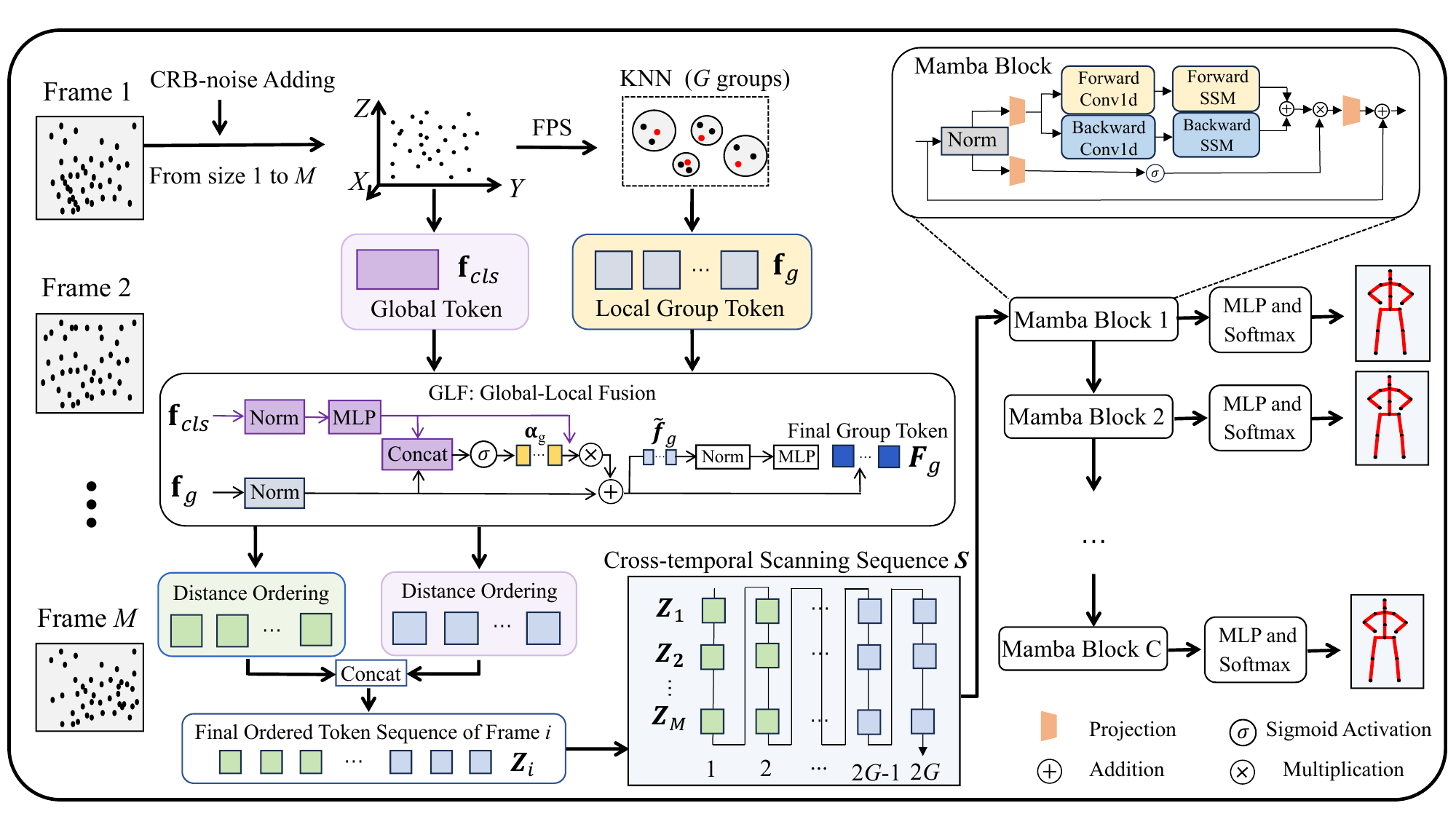}
  \caption{Overview of the ET-Mamba framework for mmWave-based pose prediction.}
  \label{fig:2}
  \vspace{-1.8em}
\end{figure*}

We model the human body as an extended target (ET) and denote its visible surface at time slot \(t\) by \(\mathcal{S}_{\text{visible},t}\).
It is partitioned into \(P\) non-overlapping regions \(\{\mathcal{S}_{p,t}\}_{p=1}^{P}\) satisfying $\mathcal{S}_{\text{visible},t} = \bigcup_{p=1}^P \mathcal{S}_{p,t}$ and $\mathcal{S}_{p_1,t} \cap \mathcal{S}_{p_2,t} = \varnothing$ for all $p_1 \neq p_2$. The echo signals are modeled as the superposition of signals scattered from the visible surface elements, while the contribution from internal reflections is negligible due to severe penetration loss. {Due to human motion and self-occlusion, the visible surface and the corresponding effective scatterers may vary across time slots. Following~\cite{10628004}, each patch on the currently visible surface is modeled by a single scatterer, which can be expressed as:}
\begin{equation}
\mathbf e[t] = \int_{\mathcal{S}_{\text{visible},t}} \mathbf e_{p}[t] \, d\rho 
\;\approx\; \sum_{p=1}^{P} \mathbf e_{p}[t],
\label{eq:echo_sum}
\end{equation}
where $e_{p}[t]$ denotes the echo from the $p$-th scatterer at time slot $t$.
By superposing all reflected components, the overall echo signal is expressed as
\begin{align}
\small
\mathbf {e}[t]
&= \sum_{p=1}^{P}
\beta_{p,t} \zeta
\mathbf{a}_r(\theta_{p,t}, \phi_{p,t})
\mathbf{a}_t^H(\theta_{p,t}, \phi_{p,t})
\mathbf{w}_r \mathbf{s}_r\!\left(t-\tfrac{2d_{p,t}}{c}\right)
,
\label{eq:total_echo}
\end{align}
where $\beta_{p,t}$ is the path-loss coefficient at time slot $t$, $\zeta$ is the complex reflection coefficient, $d_{p,t}$ is the propagation distance at time slot $t$, 
$\theta_{p,t}, \phi_{p,t}$ denote the corresponding angles at time slot $t$, respectively,  and $c$ is the speed of light.

We assume that the server mitigates the effect of static clutter via background suppression, and then applies Constant False Alarm Rate (CFAR) detection to identify the scattering points and eliminate isolated outlier detections.
Finally, fast Fourier transform (FFT)-based processing is applied to the received echo signals to extract the range, Doppler, and angle information of the scattering points. Based on the extracted information, the corresponding mmWave point-cloud features are constructed, including spatial coordinates and radial velocity.

\section{CRB-Guided Pose Prediction}
In this section, we present the proposed CRB-guided pose prediction framework for mmWave point-cloud sequences. The framework consists of two main components: an ET-Mamba-based pose prediction model and a CRB-guided anisotropic perturbation strategy for robustness enhancement. The ET-Mamba model exploits historical mmWave point-cloud frames to predict future human joint positions, while the CRB-guided perturbation strategy emulates sensing-quality variations caused by different resource allocation conditions.

\subsection{ET-Mamba Pose Prediction Model}
We adopt Mamba as the backbone of the pose prediction model. Mamba is built upon a selective state-space model (SSM) for sequence modeling. By replacing the quadratic self-attention with state updates, Mamba enables linear-time modeling for long sequences~\cite{chen2025stpm}. However, Mamba requires a one-dimensional token sequence as input, while mmWave point clouds are inherently unordered and sparse. Therefore, a naive serialization may weaken the geometric relationships among point groups. To preserve point-cloud structures while exploiting Mamba's sequential modeling ability, we propose ET-Mamba, whose overall framework is illustrated in Fig.~\ref{fig:2}.

{The model takes $\mathbf X=\{\mathbf X_i\}_{i=1}^{M}$ as input, where each point-cloud frame $\mathbf X_i$ contains $N_i$ points, and predicts the 3D human joint positions in the next time slot.}

Given the input point-cloud frame \(\mathbf{X}_i\), CRB-guided noise is first introduced to simulate sensing uncertainty caused by limited sensing resources. The detailed CRB-guided noise modeling is presented in Section~\ref{sec:crb_noise}.
To capture the global information of \(\mathbf{X}_i\), a global feature token \(\mathbf{f}_{cls}\) is extracted from the whole frame using a 1D convolutional layer followed by max pooling. For local geometric representation, $G$ group centers are selected by farthest point 
sampling (FPS), and the neighborhood of each center is constructed using $K$-nearest neighbors (KNN). Each local neighborhood is represented by center-normalized relative coordinates and velocity features, and is further encoded by a PointNet encoder~\cite{liang2024pointmamba} to obtain the local group token \(\mathbf{f}_{g}\).

While local group tokens preserve fine-grained geometric structures, their receptive fields are restricted to local neighborhoods, providing insufficient global contextual information. This limited receptive field may restrict each local token from capturing its relationship with the overall point-cloud distribution. To address this issue, we design a global-local fusion (GLF) module, which adaptively injects global contextual information into each local group token.

In the proposed GLF module, instead of directly concatenating or simply adding 
\(\mathbf{f}_{cls}\) and \(\mathbf{f}_{g}\), an adaptive gate \(\alpha_g\) is learned from 
their joint representation to regulate the contribution of global context to each 
local group.
The gated fusion process is formulated as:
\begin{subequations}
\begin{align}
\boldsymbol{\alpha}_g &= \sigma\!\left(\mathbf{W}_{\alpha}\,[\mathbf{f}_g;\mathbf{f}_{cls}] + \mathbf{b}_{\alpha}\right), 
\\
\mathbf{f}'_{g} &= \mathbf{f}_g + \boldsymbol{\alpha}_g \odot \psi(\mathbf{f}_{cls}),
\end{align}
\end{subequations}
where $[;]$ denotes concatenation, $\sigma(\cdot)$ is the sigmoid function, \(\mathbf{W}_{\alpha}\) and \(\mathbf{b}_{\alpha}\) are the learnable weight matrix and bias vector, respectively. $\psi(\cdot)$ is an MLP layer. 

The gate \(\alpha_g\) provides group-wise control over the injected global context. 
Therefore, each local group can selectively incorporate global information 
while maintaining its own local geometric representation. After global context 
injection, the fused local token is further refined by a residual feed-forward 
layer to obtain the final group token:
\begin{equation}
\mathbf{F}_g = \mathbf{f}'_g + \psi(\mathbf{f}'_g).
\end{equation}
The proposed GLF module enriches each local group token with global contextual information, 
providing representations that contain both neighborhood-level geometry and global point-cloud context 
for subsequent geometry ordering and temporal sequence modeling.

Since Mamba operates on one-dimensional token sequences, the spatial arrangement of point-cloud tokens directly affects the quality of sequential modeling. A naive flattening operation may weaken the geometric relationships among local groups. To preserve the spatial characteristics of mmWave point clouds, a geometry distance ordering strategy is introduced before sequence modeling. 

Let \(\mathcal{G}=\{\mathbf{F}_{1},\mathbf{F}_{2},\ldots,\mathbf{F}_{G}\}\) denote the set of final group tokens after GLF. Let \(\mathcal{Q}=\{\mathbf{q}_{g}\in\mathbb{R}^{3}\}_{g=1}^{G}\) denote the set of corresponding group centers, where \(\mathbf{q}_{g}\) is the center of the \(g\)-th local group.

To obtain a stable and geometry-aware ordering, two reference points are considered.
The first reference point is the radar coordinate origin, and the group tokens are sorted in ascending order according to the Euclidean distances between their centers and the origin. The second reference point is the centroid of all group centers, and the group tokens are sorted again according to the Euclidean distances between their centers and this centroid. 

This reference point based ordering introduces explicit geometric relationships into the serialized sequence. The ordering based on the radar origin preserves the spatial distribution of point groups with respect to the sensing coordinate system, while the ordering based on the group-center centroid describes the relative organization of local groups within the observed point clouds. 

The two ordered sequences are concatenated along the token dimension to form the ordered token sequence of the \(i\)-th frame, denoted by $\mathbf{Z}_{i}\in\mathbb{R}^{2G\times d_m}$, where $2G$ is the total number of ordered tokens and $d_m$ is the token feature dimension.
After obtaining the ordered token sequence \(\mathbf{Z}_i\) for each frame, temporal serialization is required before Mamba-based modeling. Instead of arranging tokens frame by frame, point-group tokens from different time steps 
are interleaved in the serialized sequence. Specifically, the serialized sequence is constructed as:
\begin{equation}
\mathbf{S}
=
[
\mathbf{z}_{1,1},
\mathbf{z}_{2,1},
\ldots,
\mathbf{z}_{M,1},
\mathbf{z}_{1,2},
\mathbf{z}_{2,2},
\ldots,
\mathbf{z}_{M,2G}
],
\end{equation}
where \(\mathbf{z}_{i,j}\in\mathbb{R}^{d_m}\) denotes the \(j\)-th ordered group token of the \(i\)-th frame. This design places temporally related tokens closer to each other in the one-dimensional sequence, thereby enabling the sequence model to capture cross-frame motion correlations more effectively.

After serializing the original point-cloud sequence into a one-dimensional token sequence, we employ a six-layer Mamba block for sequential feature modeling~\cite{chen2025stpm}. Considering that server-side 
power consumption and latency constraints may prevent full-depth inference, a lightweight 
prediction head is attached after each Bi-Mamba layer. Let \(\mathcal{C}=\{1,\cdots,C_{\max}\}\) denote the set of candidate model inference depths. In this way, the model can produce 
pose prediction results from intermediate layers when full-layer inference is not feasible, 
providing a flexible trade-off between inference accuracy and computational cost.
\subsection{CRB-Guided Perturbation and Training}
\label{sec:crb_noise}
To improve the robustness of the prediction model under varying sensing quality, a CRB-guided anisotropic perturbation strategy is introduced during training. Since the quality of mmWave point clouds is affected by sensing conditions, such as sensing power and beamforming, perturbations are injected into the nominal point-cloud observations to emulate degradation under different sensing conditions. Motivated by the quasi-static property of indoor human motion, we assume that the channel environment and sensing quality remain relatively stable within a short time window. We adopt a CRB-based uncertainty model, where the perturbations in range, angle, and radial velocity are represented by independent Gaussian variables with variances given by the corresponding CRB expressions.

As implied in \eqref{eq:total_echo}, the echo can be expressed as a superposition of reflections from multiple scattering centers. We assume that the echoes from different scattering centers are approximately independent. Under this assumption, at each training epoch, we first determine the SNR at the beam center, which can be expressed as~\cite{quang2026diffusion}:
\vspace{-0.5em}
\begin{equation}
\mathrm{SNR}_{s,o}
=
\frac{
|\zeta|^2
\left\|
\mathbf{A}(\theta_o,\phi_o)\mathbf{w}_r
\right\|_2^2
}{
4d_o^4\sigma_z^2
},
\label{eq:SNR_s_o}
\end{equation}
{where $\mathbf A(\theta_o,\phi_o)
\triangleq
\mathbf a_r(\theta_o,\phi_o)\mathbf a_t^H(\theta_o,\phi_o)$.}
Then, for the $i$-th point, we apply a distance-dependent path-loss correction to obtain $\mathrm{SNR}_{s,i} = \mathrm{G}_i\mathrm{SNR}_{s,o}(\frac{d_o}{d_i})^{4}$, where $\mathrm{G}_i$ is the beamforming gain, $d_i$ is the range of the point and $d_o$ is the distance from the radar to the beam center. Next, we obtain the beamforming gain by the following equation:
\begin{equation}
\mathrm{G}_i = 
\frac{\left|\mathbf a_t^{\mathrm H}(\theta_o,\phi_o)\,\mathbf a_t(\theta_i,\phi_i)\right|^2}
{\|\mathbf a_t(\theta_o,\phi_o)\|^2\,\|\mathbf a_t(\theta_i,\phi_i)\|^2}
\end{equation}

We use closed-form CRB expressions to map $\mathrm{SNR}_{s,i}$ to physically consistent noise variances for range, angle, and radial velocity, and inject independent zero-mean Gaussian perturbations accordingly. Specifically, the CRB variance for range is given by~\cite{quang2026diffusion}:
\vspace{-0.5em}
\begin{equation}
\sigma_{d,i}^{2}=\frac{c^{2}}{8\pi^{2}\mathrm{SNR}_{s,i}B_r^{2}}.
\label{eq:CRB_d}
\end{equation}

The CRB for radial velocity is adopted from~\cite{dogandzic2001cramer} as:
\vspace{-0.5em}
\begin{equation}
\sigma_{v,i}^{2}=\frac{3\lambda^{2}}{32\pi^{2}T_s^{2}\,\mathrm{SNR}_{s,i}\,N_c\left(N_c^{2}-1\right)},
\label{eq:CRB_v}
\end{equation}
\vspace{0em}
where $T_s$ is the sensing duration and $N_c$ is the sensing snapshots. The CRB of the angles can be expressed~\cite{dogandzic2001cramer} as
\begin{subequations}
\begin{align}
\small
\sigma_{\theta,i}^{2} &=
\frac{\lambda^{2}}{8\pi^{2}N_{p}\,\mathrm{SNR}_{s,i}\,\kappa ^{2}\,S_x\left(\cos\phi_i\,\cos\theta_i\right)^{2}},\\
\sigma_{\phi,i}^{2} &=
\frac{\lambda^{2}}{8\pi^{2}N_{p}\,\mathrm{SNR}_{s,i}\,\kappa^{2}\left[\,S_x\left(\sin\phi_i\,\sin\theta_i\right)^{2}+S_z\left(\cos\phi_i\right)^{2}\right]},
\label{eq:CRB_jiaodu}
\end{align}
\end{subequations}
where $N_{p}$ denotes the number of snapshots used for angle estimation, $\kappa$ is $2\pi/\lambda$, $S_x$ = $N_z d_x^{2} \frac{N_x\left(N_x^{2}-1\right)}{12}$ and $S_z = N_x d_z^{2} \frac{N_z\left(N_z^{2}-1\right)}{12}$. Consequently, the injected perturbations are sampled according to $\Delta_{d,i}\sim\mathcal N(0,\sigma_{d,i}^{2})$, $\Delta_{\theta,i}\sim\mathcal N(0,\sigma_{\theta,i}^{2})$, $\Delta_{\phi,i}\sim\mathcal N(0,\sigma_{\phi,i}^{2})$, and $\Delta_{v,i}\sim\mathcal N(0,\sigma_{v,i}^{2})$. To map the sensing-parameter perturbations back to the point-cloud coordinate domain, we transform the perturbed range and angular parameters into Cartesian coordinates. Let \(\mathbf{p}_i=(x_i,y_i,z_i)\) denote the initial coordinate of the \(i\)-th point. The perturbed coordinate $\mathbf{p}_{i}^{\prime}$ is obtained by:
\begin{equation}
\mathbf{p}_{i}^{\prime} = f(d_i + \Delta_{d,i}, \theta_i + \Delta_{\theta,i}, \phi_i + \Delta_{\phi,i}),
\end{equation}
where \(f(\cdot)\) denotes the spherical-to-Cartesian transformation, given by
$f(d,\theta,\phi)= [d\sin\theta\cos\phi,
d\cos\phi\cos\theta,d\sin\phi]^T
$.

Finally, the augmented feature of the \(i\)-th point is expressed as
$\mathbf{q}'_i = \left[\mathbf{p}'_i, v'_i\right]$, where $v'_i = v_i+\Delta_{v,i}$.

During training, the original mmWave point-cloud data are segmented using a sliding window to generate overlapping point-cloud sequence samples. Each sample consists of $M$ consecutive point-cloud frames and is used to predict the 3D human joint positions in the next frame. The network is optimized using Adam~\cite{adhikari2024misleep}. 
{Each target frame contains $J$ human joints, and the predicted and ground-truth positions of joint $j\in\mathcal J$ in the $b$-th sample are denoted by $\hat{\mathbf p}_{M+1,j}^{(b)}$ and $\mathbf p_{M+1,j}^{*(b)}$, respectively.} Let $N_B$ denote the batch size and $b$ denote the $b$-th clip in the batch. The mean squared error (MSE) is adopted as the training objective:
\begin{equation}
\mathcal{L}_{MSE}=\frac{1}{N_BJ}\sum_{b=1}^{N_B}\sum_{j=1}^{J}\bigl\|\hat{\mathbf{p}}^{(b)}_{M+1,j}-\mathbf{p}^{*(b)}_{M+1,j}\bigr\|_2^2.
\end{equation}


\section{PERFORMANCE ANALYSIS AND PROBLEM FORMULATION}
In the optimization stage, the AP jointly determines the beamforming matrix, model inference depth, and computation frequency according to the communication QoS requirements and sensing performance demand, enabling adaptive coordination among sensing, communication, and computation. In particular, to support stable sensing and continuous pose tracking in future time slots, the pose prediction error, measured by MPJPE, is selected as the optimization objective in the resource allocation process. 
\subsection{Energy Consumption and Latency}
The energy of the considered ISCC system consists of three parts: sensing transmission energy $E_s$, communication transmission energy $E_c$, and local computation energy $E_{\mathrm{comp}}$. 

For the local computation, the computational energy cost consists of two parts: point-cloud generation energy $E_\mathrm{pc}$ and model inference energy $E_\mathrm{in}$, i.e., \(E_{\mathrm{comp}} = E_\mathrm{pc} + E_\mathrm{in}\). For point-cloud generation, the main computation cost is caused by the FFT operation~\cite{chen2025sensing}. Specifically, in each sensing frame, the estimation requires an $N_s$-point FFT for each of the $N_c$ sensing snapshots. Therefore, the total energy cost of point-cloud generation for each sensing frame can be expressed as
\begin{equation}
    E_{\mathrm{pc}} = \gamma N_c N_s C_{pc}f_{\mathrm{pc}}^2{\log_2N_s},
\end{equation}
where $C_{pc}$ represents the computational intensity of FFT operation (in CPU cycles/element) and $f_{\mathrm{pc}}$ (in CPU cycles/s, i.e., Hz) denotes the computation resource of the server for point-cloud generation. The latency of point-cloud generation can be expressed as
\begin{equation}
    \tau_{\mathrm{pc}} = \frac{N_c N_s C_{pc}{\log_2N_s}}{f_{\mathrm{pc}}}.
\end{equation}
Let $C_{\rm glo}$, $C_{\rm grp}$, and $C_{\rm pred}$ denote the average CPU cycles required for global point-cloud processing, group feature construction, and final prediction head, respectively. Furthermore, let $C_e=C_{\rm glo}+C_{\rm grp}+C_{\rm pred}$ denote the total CPU cycles for the depth-independent encoding and decoding stages, while $C_L$ denotes the CPU cycles required for each depth layer. The energy cost of the inference can be expressed as~\cite{ding2026energy}:
\begin{equation}
    E_{\mathrm{in}} = \gamma M(C_e + CC_L)f_{\mathrm{in}}^2.
\end{equation}
The corresponding inference latency is given by:
\vspace{-0.3em}
\begin{equation}
    \tau_{\mathrm{in}} = \frac{M(C_e + CC_L)}{f_{\mathrm{in}}}.
\end{equation}
\vspace{-0.3em}
With the duration of one timeslot $T$, and the duration of sensing in one timeslot $T_s$, the sensing and communication transmission energy can be expressed as:
\begin{equation}
E_{\mathrm{s}} = T_s \|\mathbf{w}_r\|_2^2,\\ \quad
E_{\mathrm{c}} = T \sum_{u=1}^{U} \|\mathbf{w}_{c,u}\|_2^2.
\end{equation}

\subsection{Communication and sensing QoS}
 {In this paper, the communication users are assumed to be quasi-static within each time slot, and their locations are assumed to be accurately updated between consecutive slots using standard tracking methods, such as EKF. For the sensing target, the AP performs a rapid scan at task initialization to obtain an initial location estimate from the received echoes, which is used to initialize subsequent sensing beamforming and continuous tracking.} Based on that, the communication rate $R_u$ of user $u$ should satisfy $R_u \;\geq\; R_{\min}$.

To obtain accurate and complete human joint positions, the sensing QoS is constrained from two aspects: spatial coverage and echo detectability. First, a human-joint coverage constraint is introduced to ensure that all the human target joints can be effectively covered by the sensing beam. Then, an echo detectability threshold constraint is imposed to guarantee that the echo reflected from at least the human body center can be reliably detected. {Meanwhile, the sensing power across the visible scatterers should be sufficiently balanced to avoid excessive beam concentration on only a small portion of the human body}. Based on that, the beam coverage constraint of the sensing side is imposed as~\cite{wang20263d}:
\begin{equation}
    \eta \min_{1\leq j \leq J} 
    \big(\mathbf{a}_j^H \mathbf{R}_x^r \mathbf{a}_j\big)
    - \max_{1\leq j \leq J}
    \big(\mathbf{a}_j^H \mathbf{R}_x^r \mathbf{a}_j\big) \;\geq\; 0 ,
\label{coverage constraint}
\end{equation}
where $\eta$ denotes the coverage factor, $\mathbf{a}_j$ is the direction vector of $j$-th joint with $j\in\mathcal{J}$, and {$\mathbf{R}_x^r = \mathbf w_r \mathbf w_r^H$} is the sensing signal covariance matrix. {The minimum and maximum terms in (22) correspond to the weakest and strongest illumination powers among all human joints, respectively}. Let $\mathrm{SNR}_r$ denote the minimum detection threshold, the SNR of the human center should satisfy $\mathrm{SNR}_{s,o}\ge \mathrm{SNR}_r$.
\subsection{Empirical Model of MPJPE}
The MPJPE generally decreases as the model inference depth increases, since a deeper model usually provides a stronger feature-extraction capability.  Meanwhile, the achievable prediction accuracy is also limited by the quality of the sensing observations. A lower sensing SNR usually leads to less reliable input features, such as point-cloud perturbations, which results in a larger prediction error. Hence, the MPJPE can be modeled as a monotonically decreasing function of both the sensing SNR and the model inference depth, denoted by $\text{MPJPE} = m(\mathrm{SNR}_{s,o}, C)$. It is difficult to derive a closed-form analytical expression for this relationship because their impacts are generally non-linear with possible saturation effects. Therefore, an empirical fitting model is established based on experimental data. Specifically, a variant of the sigmoid function is adopted to capture the nonlinear and saturating effects of sensing SNR and model inference depth on MPJPE. By averaging the results over multiple experiments, the relationship of MPJPE with respect to sensing SNR and model inference depth can be given by:
\begin{align}
\label{M}
m(& \mathrm{SNR}_{s,o}, C)
= a
+ \frac{b}{1+e^{\lambda_1\left(\mathrm{SNR}_{s,o}-s_0\right)}}\nonumber \\
&+ \frac{d}{1+e^{\lambda_2\left(C-c_0\right)}}
+ \frac{g}{1+e^{\left(\lambda_3\,\mathrm{SNR}_{s,o}+\lambda_4C-\mu\right)}},
\end{align}
where all the model parameters, $a$, $b$, $d$, $g$, $\lambda_1$, $\lambda_2$, $\lambda_3$, $\lambda_4$ and $\mu$ are estimated via non-linear least squares fitting. {The fitted parameters are
$a=4.01$,
$(b,\lambda_1,s_0)=(4.28,0.19,-10.00)$,
$(d,\lambda_2,c_0)=(10.00,0.98,-0.25)$,
and
$(g,\lambda_3,\lambda_4,\mu)=(10.00,0.14,1.02,-1.82)$,
corresponding to the constant, SNR, inference depth,
and SNR-depth coupling terms, respectively. As illustrated in Fig.~\ref{fig:fitted_curves}, this model closely captures the
actual relationship among sensing SNR, model inference depth and MPJPE. The resulting fitting accuracy is evaluated using \(R^2\), root mean square error (RMSE), mean absolute error (MAE), and maximum error (MAXE), which are 0.9995, 0.02~cm, 0.02~cm, and 0.05~cm, respectively.}

\begin{figure}[t]
    \centering
    \begin{subfigure}[t]{0.54\linewidth}
        \centering
        \includegraphics[
            width=\linewidth,
            height=1\linewidth,
            trim=8 4 4 4,
            clip
        ]{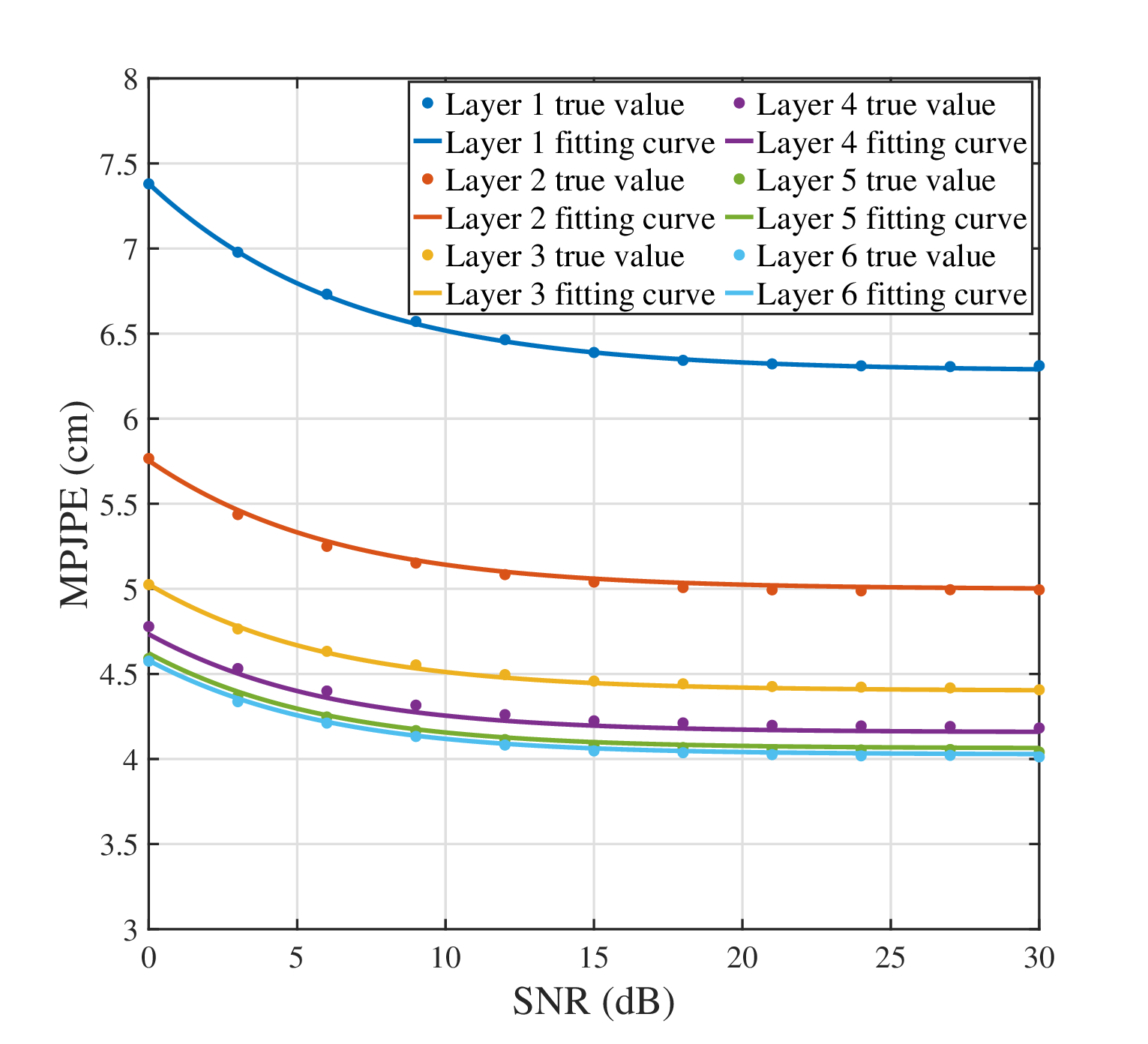}
        \caption{2D fitted curves.}
        \label{fig:curve2D}
    \end{subfigure}
    \hspace{-0.075\linewidth}
    \begin{subfigure}[t]{0.49\linewidth}
        \centering
        \includegraphics[
            width=\linewidth,
            height=1.05\linewidth,
            trim=10 4 18 4,
            clip
        ]{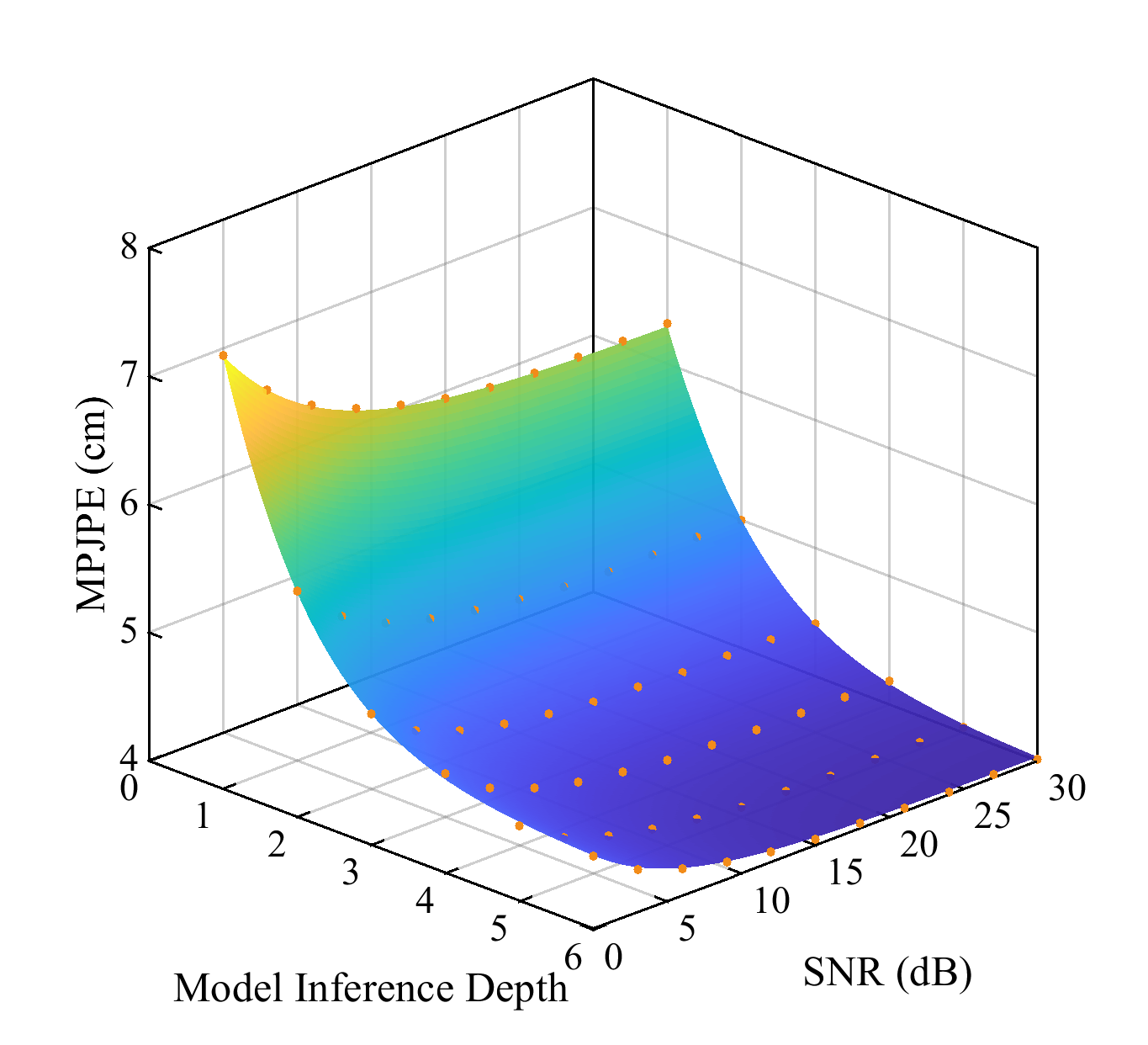}
        \caption{3D fitted surface.}
        \label{fig:curve3D}
    \end{subfigure}
    \caption{2D fitted curves and 3D fitted surface of MPJPE under different SNRs and model inference depths.}
    \label{fig:fitted_curves}
    \vspace{-1.5em}
\end{figure}

\subsection{Problem Formulation}
In this paper, we aim to minimize the pose prediction error of the ISCC system. To this end, the fitted function is used to represent the final performance of the considered system. Therefore, the pose prediction error minimization problem can be formulated as
\begin{subequations}
\label{eq:optimization_problem}
\begin{align}
    \min_{\mathbf{w}_{c,u}, \mathbf{w}_r, C, f_{\mathrm{pc}}, f_{\mathrm{in}}} & \quad m \left( \mathrm{SNR}_{s,o}, C \right), \tag{\theequation} \\
    \text{s.t.} \quad & \tau_{\mathrm{pc}} + \tau_{\mathrm{in}} \leq T,\label{tau} \\
    & E_{s} +E_{c}+ E_{\mathrm{pc}} + E_{\mathrm{in}} \leq {P}_{\max} T, \label{energy}\\
    & R_u \geq R_{\min}, \quad \forall u \in \mathcal{U}, \label{com QoS}\\
    & \mathrm{SNR}_{s,o} \geq \mathrm{SNR}_r, \label{sen QoS}\\
    & 1 \leq C \leq C_{\max}, \quad C \in \mathcal{C},\label{C}\\ 
    & 0 \leq  f_{\mathrm{pc}} \leq f_{\max}, \quad 0 \leq  f_{\mathrm{in}} \leq f_{\max},\label{f}\\
    & \tau_{\mathrm{pc}} \geq T_s, \label{h}\\
    & \eqref{coverage constraint},\label{cover}
\end{align}
\end{subequations}
where constraint \eqref{tau} ensures that point-cloud generation and model inference can be completed within one time slot. Constraint \eqref{energy} imposes the energy budget of the ISCC system. Constraint \eqref{com QoS} guarantees the communication QoS of each user. Constraint \eqref{sen QoS} ensures the minimum sensing SNR required for reliable echo detection at the human body center. Constraints \eqref{C} and \eqref{f} specify the feasible ranges of the model inference depth and computation frequency, respectively. Constraint {\eqref{h} ensures that the complete point-cloud can only be obtained after all sensing echoes within the time slot have been collected.} 
Constraint (24h) imposes the sensing beam coverage requirement over the human joints.

\section{JOINT RESOURCE ALLOCATION ALGORITHM}
Problem \eqref{eq:optimization_problem} is difficult to solve directly due to the integer variable and the non-convex objective function and constraints. To solve the optimization of beamforming matrix, model inference depth and computing frequency, we propose an alternating optimization scheme. Initially, we characterize the relationship among the sensing SNR, model inference depth and prediction error. Subsequently, the optimization problem is decomposed into two iterative subproblems. One optimizes the computation resources with fixed beamforming matrices, while the other optimizes the beamforming matrices with fixed computation resources. By alternately optimizing the two subproblems, the proposed AO-based scheme converges to a stationary solution 
\((C^*, \mathbf{w}_r^*, \{\mathbf{w}_{c,u}^*\}_{u=1}^{U}, f_{\mathrm{pc}}^*, f_{\mathrm{in}}^*)\), 
which minimizes the pose prediction error.

\subsection{AO-Based Algorithm}
The proposed iterative algorithm consists of two subproblems in each iteration. 
The first subproblem optimizes the model inference depth \(C\) and the CPU frequencies \(f_{\mathrm{pc}}\) and \(f_{\mathrm{in}}\) with fixed beamforming variables (\(\mathbf{w}_r\), \(\{\mathbf{w}_{c,u}\}_{u=1}^{U}\)), while the second subproblem optimizes \(\mathbf{w}_r\) and \(\{\mathbf{w}_{c,u}\}_{u=1}^{U}\) based on the obtained ($C$, $f_{\mathrm{pc}}$, $f_{\mathrm{in}}$) in the previous step.

\subsubsection{Computation Resource Allocation With Fixed Beamforming Matrix}
In this subproblem, since the sensing error function m(C) is monotonically decreasing with respect to the depth $C$, the optimal solution to this subproblem is the largest feasible $C$. To find the optimal $C$, we relax the discrete constraint \eqref{C} as $C \in [1, C_{\max}]$.

Accordingly, with fixed $\mathbf w_{c,u}$ and $\mathbf w_r$, the relaxed computation block is formulated as:
\begin{subequations}
\label{eq:optimization_sub1problem}
\begin{align}
    \max_{C,\mathbf f} \quad &  C \tag{\theequation} \\
    \text{s.t.} \quad & \frac{\kappa_{\mathrm{pc}}}{f_{\mathrm{pc}}} + \frac{\kappa_{\mathrm{in}}(C)}{f_{\mathrm{in}}} \leq T, \label{tausub1} \\
    &\gamma \kappa_{\mathrm{pc}}f_{\mathrm{pc}}^2 + \gamma \kappa_{\mathrm{in}}(C)f_{\mathrm{in}}^{2} \leq \hat E, \label{energysub1} \\
    & 1 \leq C \leq C_{\max},\quad C\in\mathbb{R},\label{Csub1}\\ 
    & \eqref{f},\quad \eqref{h} \nonumber
\end{align}
\end{subequations}
where $\kappa_{\mathrm{pc}}$ = $N_c N_s C_{pc}{\log_2N_s}$, $\kappa_{\mathrm{in}}(C)$ = $ M(C_e + CC_L)$ and $\hat{E}$ denotes $P_{\max}T-E_{s}-E_{c}$.

\begin{lemma}
    For the computation subproblem with a fixed beamforming matrix, the latency constraint is always active at the optimal solution, which satisfies:
\begin{equation}
\label{fc1}
    \frac{\kappa_{\mathrm{pc}}}{f_{\mathrm{pc}}^*} + \frac{\kappa_{\mathrm{in}}(C^*)}{f_{\mathrm{in}}^*} = T.
\end{equation}

\end{lemma}

\begin{proof}
Assume that at an optimal solution $(C^*, f_{\mathrm{pc}}^*, f_{\mathrm{in}}^*)$, the latency constraint is inactive such that $\kappa_{\mathrm{pc}}/f_{\mathrm{pc}}^* + \kappa_{\mathrm{in}}(C^*)/f_{\mathrm{in}}^* < T$. Since the left-hand side is continuous and increases as frequencies decrease, there exists a sufficiently small perturbation that reduces $f_{\mathrm{pc}}^*$ or $f_{\mathrm{in}}^*$ while maintaining feasibility. However, the computational energy $\gamma \kappa_{\mathrm{pc}} f_{\mathrm{pc}}^2 + \gamma \kappa_{\mathrm{in}}(C) f_{\mathrm{in}}^2$ is strictly monotonically increasing with $f_{\mathrm{pc}}$ and $f_{\mathrm{in}}$. Thus, reducing these frequencies would strictly decrease energy consumption, releasing more budget to further increase the depth $C^*$. This contradicts the optimality of $(C^*, f_{\mathrm{pc}}^*, f_{\mathrm{in}}^*)$. Consequently, the latency constraint must be active at the optimal point, i.e., $\kappa_{\mathrm{pc}}/f_{\mathrm{pc}}^* + \kappa_{\mathrm{in}}(C^*)/f_{\mathrm{in}}^* = T$.
\end{proof}
Based on \textbf{Lemma 1}, the latency constraint is active at the optimal point. Therefore, the model inference depth \(C\) can be expressed as a function of the computation frequencies \(\mathbf f=\{f_{\mathrm{pc}},f_{\mathrm{in}}\}\), which is given by:
\begin{equation}
\label{couple}
    C(f_{\mathrm{pc}}, f_{\mathrm{in}}) = \frac{f_{\mathrm{in}}(T f_{\mathrm{pc}} - \kappa_{\mathrm{pc}}) - M C_e f_{\mathrm{pc}}}{M C_L f_{\mathrm{pc}}}.
\end{equation}
By substituting \(C(f_{\mathrm{pc}},f_{\mathrm{in}})\) into \eqref{eq:optimization_sub1problem}, the computation resource subproblem can be rewritten as:
\begin{subequations} \label{eq:reformulated_problem1}
\begin{align}
    \max_{f_{\mathrm{pc}}, f_{\mathrm{in}}} \quad &  \frac{f_{\mathrm{in}} (T f_{\mathrm{pc}} - \kappa_{\mathrm{pc}}) - M C_e f_{\mathrm{pc}}}{M C_L f_{\mathrm{pc}}}  \tag{\theequation} \label{fcin} \\
    \text{s.t.} \quad & \gamma \kappa_{\mathrm{pc}} f_{\mathrm{pc}}^2 + \gamma f_{\mathrm{in}}^3  \left( T - \frac{\kappa_{\mathrm{pc}}}{f_{\mathrm{pc}}} \right)  \le \hat{E}, \label{fcin2}\\
    & 1 \le \frac{f_{\mathrm{in}} (T f_{\mathrm{pc}} - \kappa_{\mathrm{pc}}) - M C_e f_{\mathrm{pc}}}{M C_L f_{\mathrm{pc}}} \le C_{\max}, \label{fcin3}\\
    & \frac{\kappa_{\mathrm{pc}}}{f_{\mathrm{pc}}} \ge T_s, \label{fcin4}\\
    & \eqref{f} \nonumber.
\end{align}
\end{subequations}
To solve the non-convex problem \eqref{eq:reformulated_problem1}, an AO-based method is proposed that iteratively optimizes $f_{\mathrm{pc}}$ and $f_{\mathrm{in}}$.

For the optimization of \(f_{\mathrm{pc}}\) with fixed \(f_{\mathrm{in}}\), since \(m(C)\) is monotonically decreasing with respect to \(C\) and \(C(f_{\mathrm{pc}},f_{\mathrm{in}})\) is monotonically increasing with respect to $f_{\mathrm{pc}}$, minimizing \(m(C)\) is equivalent to maximizing \(f_{\mathrm{pc}}\). Therefore, for fixed \(f_{\mathrm{in}}\), the optimization of \(f_{\mathrm{pc}}\) in \eqref{eq:reformulated_problem1} can be denoted as:
\begin{subequations}
\label{eq:reformulated_problem2}
\begin{align}
    \max_{f_{\mathrm{pc}}} \quad & f_{\mathrm{pc}} \tag{\theequation} \\
    \text{s.t.} \quad & \eqref{fcin2}, \eqref{fcin3}, \eqref{fcin4}, \eqref{f}. \nonumber
\end{align}
\end{subequations}
\begin{proposition}
The optimal point-cloud generation frequency $f_{\mathrm{pc}}^*$ for \eqref{eq:reformulated_problem2} is achieved at the boundary of the feasible region defined by the constraints. Specifically, the optimal solution is given by:
\begin{equation}
    f_{\mathrm{pc}}^* = \min \{ f_{\mathrm{pc}}^{E}, f_{\mathrm{pc}}^{C}, \frac{\kappa_{\mathrm{pc}}}{T_s}, f_{\max} \}
\end{equation}
where $f_{\mathrm{pc}}^{E}$, $f_{\mathrm{pc}}^{C}$, and $\frac{\kappa_{\mathrm{pc}}}{T_s}$ denote the upper bounds on $f_{\mathrm{pc}}$ induced by constraints \eqref{fcin2}, \eqref{fcin3} and \eqref{fcin4}, respectively.
\end{proposition}
\begin{proof}
The first order derivative of \eqref{fcin2} is $2 \gamma \kappa_{\mathrm{pc}} f_{\mathrm{pc}} + \frac{\gamma f_{\mathrm{in}}^3 \kappa_{\mathrm{pc}}}{f_{\mathrm{pc}}^2}$, which is positive and monotonically increasing with respect to \(f_{\mathrm{pc}}\). This imposes an upper bound $f_{\mathrm{pc}}^{E}$ on $f_{\mathrm{pc}}$. Similarly, the first order derivative of \eqref{fcin3} is $\frac{f_{\mathrm{in}} \kappa_{\mathrm{pc}}}{M C_L f_{\mathrm{pc}}^2}$, the upper bound is denoted as $f_{\mathrm{pc}}^{C}$. Therefore, $f_{\mathrm{pc}}^*$ is the minimum of all such individual upper bounds to ensure all constraints are simultaneously satisfied.
\end{proof}

Similarly, for the optimization of \(f_{\mathrm{in}}\) with fixed \(f_{\mathrm{pc}}\), $C(f_{\mathrm{pc}}, f_{\mathrm{in}})$ is monotonically increasing with $f_{\mathrm{in}}$. Therefore, the optimization subproblem of \(f_{\mathrm{in}}\) with fixed $f_{\mathrm{pc}}$ is given by:
\vspace{-1em}
\begin{subequations}
\label{eq:reformulated_problem3}
\begin{align}
    \max_{f_{\mathrm{in}}} \quad & f_{\mathrm{in}} \tag{\theequation} \\
    \text{s.t.} \quad & \eqref{fcin2}, \eqref{fcin3}, \eqref{f}.\nonumber
\end{align}
\end{subequations}

Similarly to \eqref{eq:reformulated_problem2}, the relevant constraints in \eqref{eq:reformulated_problem3} are monotonically increasing with respect to $f_{\mathrm{in}}$, the optimal solution $f_{\mathrm{in}}^*$ is the maximum feasible value that satisfies all constraints simultaneously. The closed-form solution is given by:
\begin{equation}
    f_{\mathrm{in}}^* = \min \left\{ f_{\mathrm{in}}^{E}, f_{\mathrm{in}}^{C}, f_{\max} \right\},
\end{equation}

where \(f_{\mathrm{in}}^{E}\) and \(f_{\mathrm{in}}^{C}\) denote the upper bounds on \(f_{\mathrm{in}}\) induced by the energy constraint \eqref{fcin2} and the model-depth constraint \eqref{fcin3}, respectively.

After the AO updates of \(f_{\mathrm{pc}}\) and \(f_{\mathrm{in}}\), a continuous inference depth $\tilde{C}$ is obtained according to \eqref{couple}. Then, $\tilde{C}$ is projected onto the feasible set $\mathcal{C}$ by selecting the largest feasible integer $C^*$ that satisfies all constraints. In the 
proposed algorithm, a feasible initialization is adopted such that at least the 
minimum inference depth \(C=1\) can be supported. Once the discrete depth \(C^*\) is determined, the CPU frequencies \(f_{\mathrm{pc}}\) and \(f_{\mathrm{in}}\) are re-optimized to satisfy \eqref{fc1} while minimizing the computation energy $E_{\mathrm{comp}}$. By minimizing $E_{\mathrm{comp}}$ for the fixed $C^*$, the remaining budget can then be reallocated to beamforming matrices in subsequent iterations of the alternating optimization.
Thus, the CPU frequency $f_{\mathrm{pc}}$ and $f_{\mathrm{in}}$ re-optimization problem with fixed $C^*$ can be expressed as:

\begin{subequations}
\begin{align}
\min_{f_{\mathrm{pc}},\,f_{\mathrm{in}}} \quad
& \gamma\kappa_{\mathrm{pc}}f_{\mathrm{pc}}^{2}
+\gamma\kappa_{\mathrm{in}}(C^*)f_{\mathrm{in}}^{2} \tag{\theequation} \\
\mathrm{s.t.}\quad
& \eqref{tausub1}, \eqref{f}.\nonumber
\end{align}
\end{subequations}

\begin{proposition}
The optimal $f_{\mathrm{pc}}^*$ and $f_{\mathrm{in}}^*$ with fixed $C^*$ can be expressed as:
\begin{equation}
\label{eq:optimal_freq}
f_{\mathrm{pc}}^{*}
=
\min\left\{
\frac{\kappa_{\mathrm{pc}}+\kappa_{\mathrm{in}}(C^*)}{T},
\frac{\kappa_{\mathrm{pc}}}{T_s}
\right\},
\quad
f_{\mathrm{in}}^{*}
=
\frac{\kappa_{\mathrm{in}}(C^*)}
{T-\frac{\kappa_{\mathrm{pc}}}{f_{\mathrm{pc}}^{*}}}.
\end{equation}
\end{proposition}
\begin{proof}
The Lagrangian of the above problem is given by:
$
\mathcal{L}
=
\gamma \kappa_{\mathrm{pc}} f_{\mathrm{pc}}^{2}
+\gamma \kappa_{\mathrm{in}} f_{\mathrm{in}}^{2}
+\lambda
\left(
\frac{\kappa_{\mathrm{pc}}}{f_{\mathrm{pc}}}
+\frac{\kappa_{\mathrm{in}}}{f_{\mathrm{in}}}
-T
\right)
+\nu
\left(
f_{\mathrm{pc}}
-\frac{\kappa_{\mathrm{pc}}}{T_s}
\right)$,
where \(\lambda\) is the multiplier associated with the active latency constraint and \(\nu \geq 0\) is the multiplier associated with the upper bound \(f_{\mathrm{pc}}\leq \kappa_{\mathrm{pc}}/T_s\).
The Karush-Kuhn-Tucker (KKT) stationarity conditions are
$2\gamma\kappa_{\mathrm{pc}} f_{\mathrm{pc}}
-\lambda\frac{\kappa_{\mathrm{pc}}}{f_{\mathrm{pc}}^{2}}
+\nu =0,\quad 2\gamma \kappa_{\mathrm{in}} f_{\mathrm{in}}
-\lambda\frac{\kappa_{\mathrm{in}}}{f_{\mathrm{in}}^{2}}
=0$. Thus, the optimal solution can be obtained from the following two cases:

1) When \(\nu=0\), the upper-bound constraint is inactive. From the stationarity conditions, we have \(f_{\mathrm{pc}}=f_{\mathrm{in}}\). Substituting this result into the latency equality constraint gives $f_{\mathrm{pc}}=f_{\mathrm{in}}=
\frac{\kappa_{\mathrm{pc}}+\kappa_{\mathrm{in}}(C^*)}{T}$.

2) When \(\nu>0\), the upper-bound constraint is active. Thus, $f_{\mathrm{pc}}=\frac{\kappa_{\mathrm{pc}}}{T_s},$
and the corresponding \(f_{\mathrm{in}}\) is obtained from the latency equality constraint as $f_{\mathrm{in}}
=
\frac{\kappa_{\mathrm{in}}(C^*)}
{T-\frac{\kappa_{\mathrm{pc}}}{f_{\mathrm{pc}}}}.
$

Combining the above two cases, the optimal CPU frequencies are given by \eqref{eq:optimal_freq}.
\end{proof}

\subsubsection{Beamforming Matrix Design With Fixed Computation Resources}
Let $\tilde E_{\mathrm{comp}}$ denote the computation energy obtained from ($f_{\mathrm{pc}}$, $f_{\mathrm{in}}$ and $C)$ in the previous step.
After solving the computation resource subproblem, we begin the subproblem of beamforming matrix design with fixed ($f_{\mathrm{pc}}$, $f_{\mathrm{in}}$ and $C)$. Let $\mathbf{R}_x^{\mathrm{c}}$ denote $\sum_{u=1}^{U}\mathbf{w}_{\mathrm{c},u}\mathbf{w}_{\mathrm{c},u}^{H}$, the subproblem about beamforming vector $\mathbf w_c$ and $\mathbf w_r$ can be expressed as:
\begin{subequations}
\begin{align}
\label{aggrewcwr}
\max_{\mathbf{w}_r,\mathbf{w}_c}\quad
& \mathrm{SNR}_{s,o} \tag{\theequation} \\
\mathrm{s.t.}\quad
& \mathrm{Tr}(\mathbf{R}_x^c)T+\mathrm{Tr}(\mathbf{R}_x^r) T_s+\tilde E_{\mathrm{comp}} \le P_{{max}}T, \\
& \eqref{com QoS}, \eqref{sen QoS}, \eqref{coverage constraint}.\nonumber 
\end{align}
\end{subequations}

To solve the beamforming design subproblem, semidefinite relaxation (SDR) is adopted by defining \( \mathbf{W}_{r} = \mathbf{w}_{r} \mathbf{w}_{r}^H \) with the constraints \( \mathbf{W}_{r} \succeq 0 \) and \( \mathbf{W}_{c,u} = \mathbf{w}_{c,u} \mathbf{w}_{c,u}^H \) with the constraints \( \mathbf{W}_{c,u} \succeq 0 \). 
Let $\mathbf A_o$ denote $ \mathbf A(\theta_o,\phi_o)\mathbf A^H(\theta_o,\phi_o)$, $\mathbf A_j$ denote $ \mathbf a_j\mathbf a^H_j$, $\mathbf{H}_u \triangleq \mathbf{h}_u\mathbf{h}_u^H$, and $\kappa_s $ denote $\triangleq \frac{|\zeta|^2}{4d_0^4\sigma_z^2}$, thus, $\mathrm{SNR}_{s,o}$ can be expressed as $\kappa_s\,\operatorname{Tr}(\mathbf W_r \mathbf A_o)$.

Accordingly, the objective function becomes 
\begin{align}
m(\mathbf W_r)
=
\bar a
+ \frac{b}{1+e^{\bar\lambda_1\mathrm{Tr}(\mathbf W_r\mathbf A_o)-\bar s_0}}
+ \frac{g}{1+e^{\bar\lambda_3\mathrm{Tr}(\mathbf W_r\mathbf A_o)-\bar\mu}} ,
\end{align}
where \(\bar a\), \(\bar s_0\), \(\bar \mu\), \(\bar\lambda_1\), and \(\bar\lambda_3\) are equivalent fitting parameters obtained after fixing \(C^*\) and substituting \(\mathrm{SNR}_{s,o}=\kappa_s\operatorname{Tr}(\mathbf W_r\mathbf A_o)\). 
Since the error function $m(\mathbf W_r)$ is monotonically decreasing with respect to the scalar quantity $\mathrm{Tr}(\mathbf W_r \mathbf A_o)$. Therefore, the optimal beamforming design is obtained by maximizing the feasible value of \(\mathrm{Tr}(\mathbf W_r\mathbf A_o)\).

\begin{algorithm}[t]
\caption{AO-Based Joint Resource Allocation Algorithm}
\label{alg:ao}
\begin{algorithmic}[1]
\State \textbf{Initialize} $C^{(0)}=1$, and obtain $f_{\rm pc}^{(0)}$ and $f_{\rm in}^{(0)}$ according to Proposition 2.
\State \textbf{Initialize} $\mathbf{w}_{r}^{(0)}$ as the minimum-power sensing beamforming vector satisfying \eqref{sen QoS}.
\State \textbf{Initialize} $\{\mathbf{w}_{c,u}^{(0)}\}_{u=1}^{U}$ as the minimum-power communication beamforming vectors satisfying \eqref{com QoS}.
\State \textbf{Input} $\mathbf{h}_u$, $\mathbf A_j$, $\{\mathbf{w}_{c,u}^{0}\}_{u=1}^{U}$, $\mathbf{w}_{r}^{0}$, $f_{\mathrm{pc}}^{0}$, $f_{\mathrm{in}}^{0}$, $C^{0}$.
\State set $I = 0$ and set the maximal error tolerance $\epsilon > 0$.
\State Define $\mathbf{W}_{c,u}=\mathbf{w}_{c,u}\mathbf{w}_{c,u}^{H}$ and $\mathbf{W}_{r}=\mathbf{w}_{r}\mathbf{w}_{r}^{H}$.
\While{not converged}
    \State Set $l=0$, $C^{(l)} =C^{(I)}$, $f_{\mathrm{pc}}^{(l)}=f_{\mathrm{pc}}^{(I)}$, $f_{\mathrm{in}}^{(l)}=f_{\mathrm{in}}^{(I)}$ 
    \State  Decouple $C^{(l)}$ from $f_{\rm pc}^{(l)}$ and $f_{\rm in}^{(l)}$ according to \eqref{couple}.
    \Repeat
        \State Update $f_{\rm pc}^{(l+1)}$ by Proposition 1 with given $f_{\rm in}^{(l)}$.
    \State Update $f_{\rm in}^{(l+1)}$ by (32) with given $f_{\rm pc}^{(l+1)}$.
    \State Set $l=l+1$.
    \Until{$\left|f_{\rm pc}^{(l+1)} - f_{\rm pc}^{(l)}\right| < \epsilon$ and $\left|f_{\rm in}^{(l+1)} - f_{\rm in}^{(l)}\right| < \epsilon$}.
    \State  Obtain the continuous depth $\tilde{C}$ according to \eqref{couple}.
    \State  Obtain \(C^{(I+1)}\) by projecting $\tilde{C}$ onto the feasible integer set.
    \State  Obtain $f_{\mathrm{pc}}^{(I+1)}$ and $f_{\mathrm{in}}^{(I+1)}$ according to Proposition 2 with fixed $C^{(I+1)}$.
\State Obtain $\{\mathbf{W}_{c,u}^{(I+1)}\}_{u=1}^{U}$ and $\mathbf{W}_{r}^{(I+1)}$ by solving \eqref{realwcwr}.
\State Set $I= I + 1$.
\EndWhile
    \State {Recover $\{\mathbf{w}_{c,u}^{\star}\}_{u=1}^{U}$ and
$\mathbf{w}_{r}^{\star}$ from
$\{\mathbf{W}_{c,u}^{\star}\}_{u=1}^{U}$ and
$\mathbf{W}_{r}^{\star}$ via Gaussian randomization.}
\State \textbf{Output} $\mathbf{w}^{\star}_{c,u}$, $\mathbf{w}^{\star}_{r}$, $C^{\star}$, $f^{\star}_{in}$, and $f^{\star}_{pc}$.
\end{algorithmic}
\end{algorithm}
The minimum communication QoS can be equivalently transformed into the following SINR threshold constraint:
\begin{align}
\mathrm{SINR}_u \ge \mathrm{SINR}_{\mathrm{min}},\quad
\mathrm{SINR}_{\mathrm{min}} = 2^{R_{\min}/B_c}-1.
\end{align}

Therefore, the communication QoS requirement reduces to a set of fixed SINR constraints. Based on that, \eqref{aggrewcwr} can be reformulated as the following semidefinite programming (SDP) problem:
\begin{subequations}
\begin{align}
\label{realwcwr}
&\max_{\{\mathbf{W}_{c,u}\},\,\mathbf{W}_r} \quad 
 \mathrm{Tr}( \mathbf{W}_r \mathbf{A}_o) \tag{\theequation} \\
\mathrm{s.t.}\quad 
& \sum_{u=1}^{U}\mathrm{Tr}(\mathbf{W}_{c,u})T +\mathrm{Tr}(\mathbf{W}_r)T_s\leq \tilde E_{\mathrm{t}}, \\
& \mathrm{Tr}(\mathbf{H}_u\mathbf{W}_{c,u})
\geq
\mathrm{SINR}_{\mathrm{min}}\left(\sum_{i\neq u}\mathrm{Tr}(\mathbf{H}_u\mathbf{W}_{c,i})+\sigma_n^2\right), \\
& \kappa_s\,\mathrm{Tr}(\mathbf{A}_o\mathbf{W}_r)\geq \mathrm{SNR}_r, \\
& \eta\mathrm{Tr}(\mathbf A_j\mathbf W_r)
\geq
\mathrm{Tr}(\mathbf A_i\mathbf W_r),
\quad  \forall i,j \in \mathcal{J}, \\
& \mathbf{W}_{c,u}\succeq 0,\quad \forall u \in \mathcal{U},\\
& \mathbf{W}_r\succeq 0,
\end{align}
\end{subequations}
where $\tilde E_{\mathrm{t}} = P_{\max}T -\tilde E_{\mathrm{comp}}$. \eqref{realwcwr} is convex and we can solve it efficiently using CVX~\cite{zhang2024joint}. {To
obtain rank-one beamformers, Gaussian randomization~\cite{wang20263d} is employed to recover feasible beamforming vectors
$\{\mathbf{w}_{c,u}\}_{u=1}^{U}$ and $\mathbf{w}_{r}$
from $\{\mathbf{W}_{c,u}\}_{u=1}^{U}$ and $\mathbf{W}_{r}$.}

With the objective of minimizing pose prediction error, the proposed AO-based algorithm can efficiently solve the coupled optimization of beamforming, model inference depth, and computation frequency.

\subsection{Algorithm Analysis}
The proposed AO-based algorithm alternately optimizes the computation resources and the beamforming matrices. In each iteration, the computation resource allocation step updates the CPU frequencies and model inference depth within the feasible region, while the beamforming design step solves an SDP problem to improve the sensing SNR. Since the fitted prediction error decreases with both model depth and sensing SNR, each update does not increase the objective value. Since MPJPE is non-negative and lower-bounded, the generated objective sequence is guaranteed to converge. In addition, the model depth is selected from a finite discrete set, so the algorithm reaches a stable solution after a finite number of iterations in practice.

\subsubsection{Complexity Analysis}
As shown in Algorithm~\ref{alg:ao}, for each AO iteration, the complexity of the computation subproblem can be treated as $\mathcal{O}(1)$. The complexity of the beamforming subproblem  can be expressed as $\mathcal{O}\big(I(U+1)^3N_t^6\big)$, where $I$ denotes the number of iterations. {For comparison, exhaustive search over the model inference depth first checks the feasibility of each candidate $C$ and solves the beamforming subproblem only for feasible candidates. Let $N_C \leq C_{\max}$ denote the number of feasible inference depths. Its computational complexity is therefore $\mathcal{O}\!\left(C_{\max}+N_C(U+1)^3N_t^6\right)$. Therefore, the proposed AO algorithm has a lower computational complexity than exhaustive search.}

\section{Results and Analysis}

\subsection{Experimental and Simulation Settings
}
\subsubsection{Experimental Settings}
The evaluation of ET-Mamba uses the \textit{MM-Fi} dataset~\cite{yang2023mm}, a publicly available multimodal dataset designed for indoor human sensing tasks. This dataset contains mmWave radar point clouds and pose annotations derived from a depth camera. The point-cloud data are segmented into $M=8$ frame sequences using a sliding-window strategy, where each sequence is used to predict the 3D human joint positions in the next frame. The generated sequence samples are then divided into training and test sets with an 8:2 ratio.

MPJPE and PA-MPJPE are adopted to quantitatively evaluate the pose prediction performance.
PA-MPJPE further aligns the predicted pose with the ground truth through Procrustes analysis, including optimal rotation, translation, and scaling, before computing the MPJPE. This metric reduces the influence of global misalignment and better reflects the structural accuracy of the predicted pose. Lower values of both MPJPE and PA-MPJPE indicate higher pose prediction accuracy.

The proposed model is implemented in PyTorch. The Adam optimizer is used with an initial learning rate of \(0.001\) and a batch size of \(128\). The model is trained for 100 epochs. For local feature grouping, the number of neighbors in KNN is set to 16, and the group size is set to 32. The input data dimension is $N_i \times 4$, and the feature dimension is set to 384. All experiments are conducted on an NVIDIA A100 GPU.

\begin{figure*}[!t]
    \centering

    \begin{subfigure}[t]{0.26\textwidth}
        \centering
        \includegraphics[width=\linewidth]{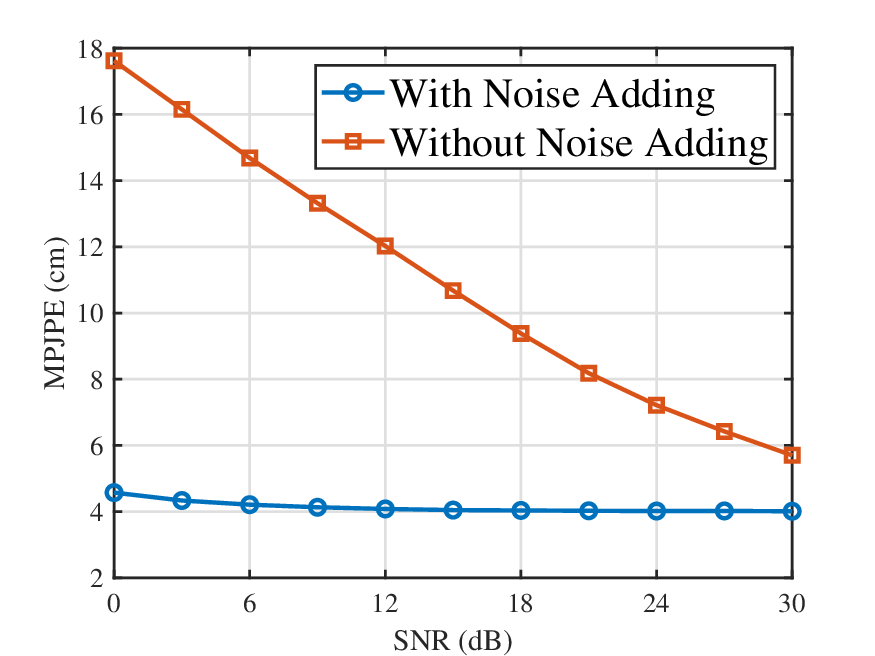}
        \caption{}
        \label{noise}
    \end{subfigure}
    \hfill
    \hspace{-0.070\linewidth}
    \begin{subfigure}[t]{0.74\textwidth}
        \centering
        \includegraphics[width=\linewidth]{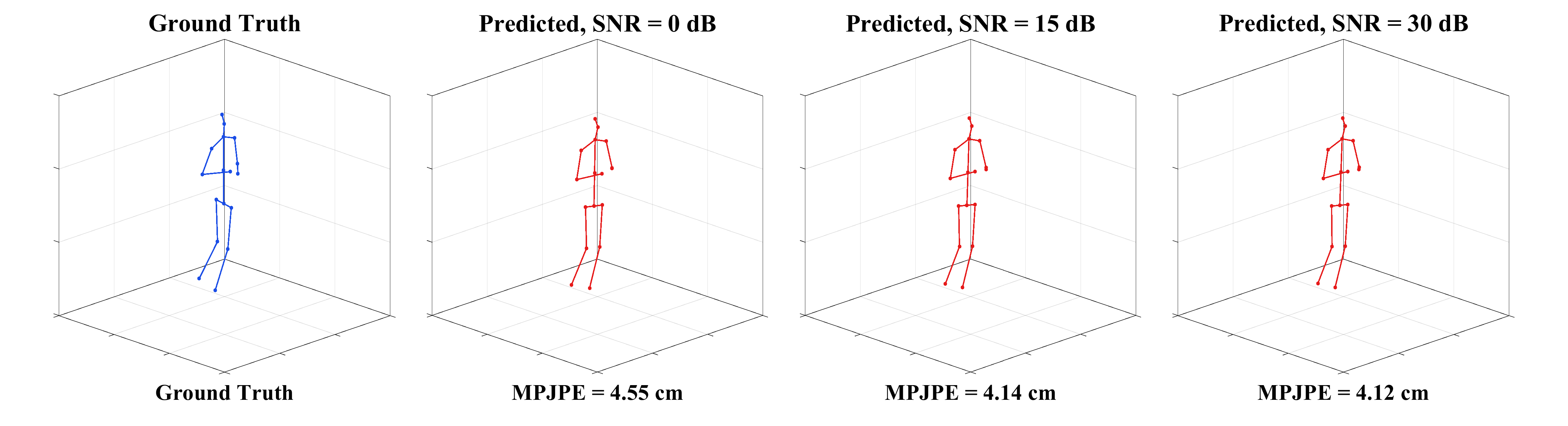}
        \caption{}
        \label{joints}
    \end{subfigure}

    \caption{Robustness evaluation and qualitative pose prediction results under different sensing SNR levels.}
\label{fig:two_subfiguresabout}
\vspace{-1em}
\end{figure*}

\subsubsection{Simulation Settings}
All system and computation parameters are summarized in Table~\ref{tab:simulation_parameters}. Unless otherwise specified, these parameters are used as the default settings in the following simulations.
\vspace{-1em}
\subsection{Performance of ET-Mamba}
For comparison, the baseline methods include \textbf{PointMamba}~\cite{liang2024pointmamba}, \textbf{PointFormer}~\cite{shi2020point}, and \textbf{STPM}~\cite{chen2025stpm}. These baselines are selected to provide a comprehensive comparison across different types of point-cloud learning architectures. Specifically, PointFormer represents attention-based methods, PointMamba represents recent Mamba-based architectures, and STPM highlights a BiMamba-based baseline for point-cloud sequences. All methods are trained and evaluated under the same experimental protocol to ensure a fair comparison. The relevant results are shown in Table~\ref{tab:performance_comparison}.
\begin{table}[t]
\centering
\caption{SIMULATION PARAMETERS}
\label{tab:simulation_parameters}
\renewcommand{\arraystretch}{1.1}
\begin{tabular}{ll}
\toprule
\textbf{Parameter} & \textbf{Value} \\
\midrule
Number of communication users, $U$ & 4 \\
UPA size, $(N_x,N_z)$ & $(4,4)$ \\
Communication noise power, $\sigma_n^2$ & $10^{-6}$ \\
Sensing noise power, $\sigma_z^2$ & $10^{-6}$ \\
Communication bandwidth, $B_c$ & $0.5~\mathrm{GHz}$ \\
Sensing bandwidth, $B_r$ & $0.5~\mathrm{GHz}$ \\
Minimum communication rate, $R_{\min}$ & $3\,\mathrm{Gbps}$ \\
Slot duration, $T$ & $0.1\,\mathrm{s}$ \\
Sensing duration, $T_s$ & $0.05\,\mathrm{s}$ \\
Minimum sensing SNR threshold, $\mathrm{SNR}_r$ & 0~dB \\
Coverage factor, $\eta$ & 3 \\
Maximum computation frequency, $f_{\max}$ & $10^9\,\mathrm{Hz}$ \\
Maximum Mamba depth, $C_{\max}$ & 6 \\
Effective switched-capacitance coefficient, $\gamma$ & $10^{-25}$ \\
Number of snapshots per frame, $N_c$ & 128 \\
FFT size, $N_s$ & 256 \\
\bottomrule
\end{tabular}
\end{table}
\begin{table}[!t]
\centering
\caption{Performance comparison of different methods for mmWave pose prediction}
\label{tab:performance_comparison}
\footnotesize
\renewcommand{\arraystretch}{1.1}
\resizebox{\columnwidth}{!}{
\begin{tabular}{lcccc}
\toprule
\textbf{Method} & \textbf{MPJPE (cm)} & \textbf{PA-MPJPE (cm)} & \textbf{Parameters (M)} & \textbf{FLOPs (G)} \\
\midrule
PointMamba & 7.79 & 4.23 & 15.8 & 2.5 \\
PointFormer & 5.54 & 3.17 & 14.1 & 16.6 \\
{STPM} & 3.85 & 2.42 & 15.9 & 3.4 \\
ET-Mamba (Ours) & \textbf{3.24} & \textbf{2.20} & \textbf{15.4} & \textbf{3.5} \\
\bottomrule
\end{tabular}
}
\end{table}
\begin{table}[!t]
\centering
\caption{Ablation study}
\label{tab:ablation}
\footnotesize
\renewcommand{\arraystretch}{1.1}
\resizebox{\columnwidth}{!}{
\begin{tabular}{lcc}
\toprule
\textbf{Method} & \textbf{MPJPE (cm)} & \textbf{PA-MPJPE (cm)} \\
\midrule
ET-Mamba (Full) & \textbf{3.24} & \textbf{2.20} \\
ET-Mamba w/o CRB noise & 3.27 & 2.21 \\
ET-Mamba w/o distance ordering (radar) & 3.58 & 2.34 \\
ET-Mamba w/o distance ordering (center) & 3.49 & 2.37 \\
ET-Mamba w/o GLF & 3.54 & 2.32 \\
\bottomrule
\end{tabular}
}
\end{table}
\vspace{0em}

\subsubsection{Quantitative Comparison}
The proposed model is compared with representative baselines in terms of MPJPE, PA-MPJPE, model parameters, and FLOPs. The results show that ET-Mamba achieves the best performance in terms of both MPJPE and PA-MPJPE, indicating its superior performance in mmWave-based pose prediction.  Compared with the strongest baseline STPM, ET-Mamba reduces MPJPE from \(3.85\)~cm to \(3.24\)~cm and PA-MPJPE from \(2.42\)~cm to \(2.20\)~cm, corresponding to relative improvements of \(15.84\%\) and \(9.09\%\), respectively. Compared with the attention-based PointFormer, ET-Mamba reduces the FLOPs by \(78.92\%\), with the MPJPE and PA-MPJPE reduced by \(41.52\%\) and \(30.60\%\), respectively. Overall, ET-Mamba achieves a favorable balance between prediction accuracy and computational complexity.

\subsubsection{Ablation Study}
As shown in Table~\ref{tab:ablation}, a comprehensive ablation study is conducted on the MM-Fi dataset to quantify the contribution of different components in the proposed ET-Mamba framework. {Compared with the variant without the multi-sort strategy,
removing the distance ordering with the radar as the reference point increases MPJPE to 3.58~cm and PA-MPJPE to 2.34~cm, while removing the distance ordering with the center as the reference point increases MPJPE to 3.49~cm and PA-MPJPE to 2.37~cm. This is because the two distance-ordering strategies provide different geometric views of the point-cloud. Specifically, the radar-based ordering describes the spatial distribution of local groups from the sensing viewpoint by measuring their distances to the radar origin, while the center-based ordering characterizes the internal geometric organization of the point-cloud by measuring their distances to the centroid of all group centers. Removing either ordering strategy reduces the geometric information preserved in the serialized token sequence, degrading the subsequent Mamba-based sequence modeling. In addition, compared with the variant without the GLF module, the full model reduces MPJPE from \(3.54\)~cm to \(3.24\)~cm and PA-MPJPE from \(2.32\)~cm to \(2.20\)~cm. This is because, without the GLF module, each local group token is mainly represented by its neighborhood-level features and cannot effectively incorporate the global contextual information of the entire point cloud, thereby degrading pose prediction performance.}
\begin{figure*}[!t]
    \centering
    \begin{subfigure}[t]{0.33\textwidth}
        \centering
        \includegraphics[width=\linewidth]{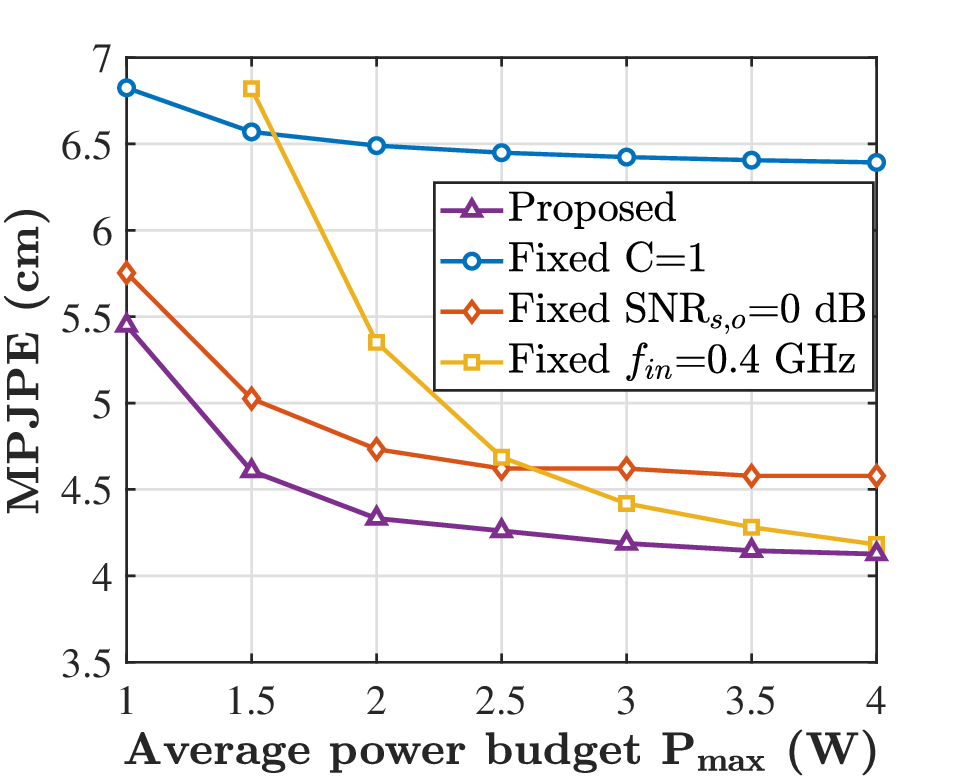}
        \caption{}
        \label{total energy}
    \end{subfigure}
    \hfill
    \begin{subfigure}[t]{0.32\textwidth}
        \centering
        \includegraphics[width=\linewidth]{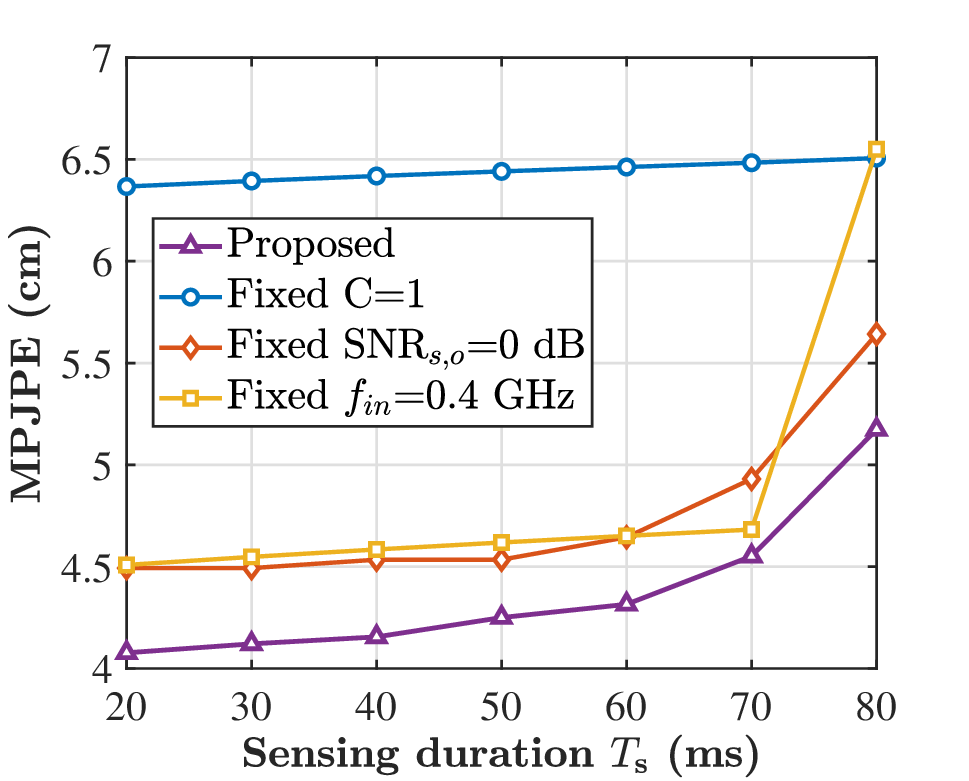}
        \caption{}
        \label{Tsensing}
    \end{subfigure}
    \hfill
    \begin{subfigure}[t]{0.33\textwidth}
        \centering
        \includegraphics[width=\linewidth]{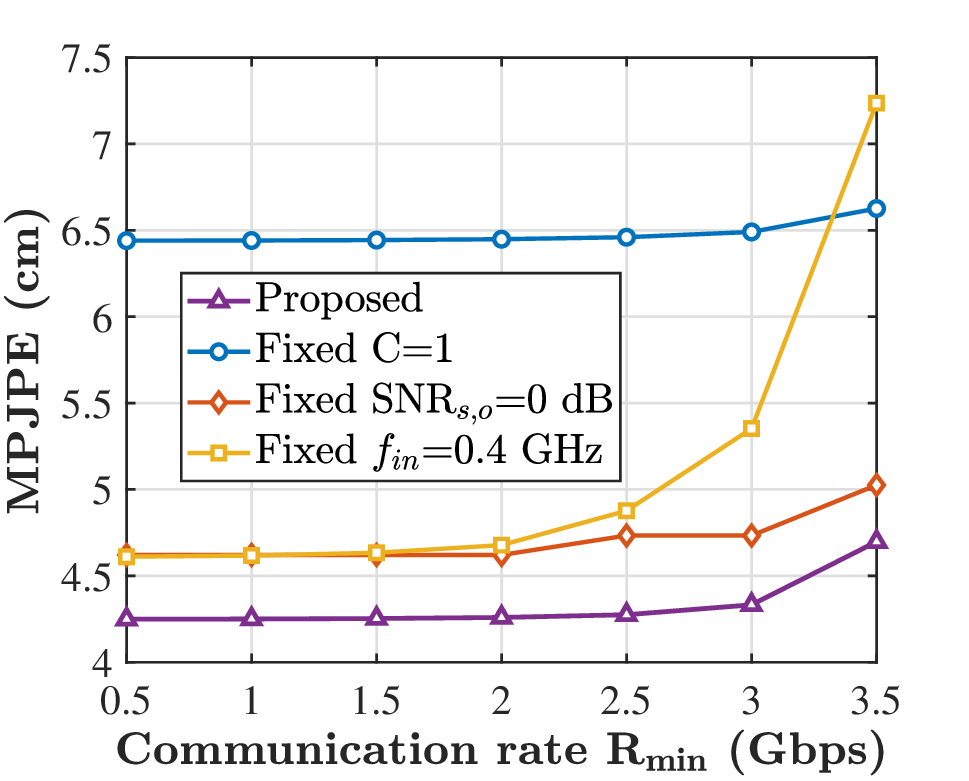}
        \caption{}
        \label{Rcommunciation}
    \end{subfigure}
    \caption{MPJPE performance versus average power budget, sensing duration, and minimum communication rate.}
    \label{fig:optEthree}
    \vspace{0em}
\end{figure*}

Although the improvement in average prediction performance is limited, the CRB-guided perturbation is designed to expose the model to sensing uncertainty during training. 

To evaluate its robustness, we compare the models trained with and without CRB-guided perturbation under different SNR levels during testing.
As shown in Fig.~\ref{fig:two_subfiguresabout}(\subref{noise}), the proposed model with CRB-guided perturbation maintains a consistently low MPJPE across the entire SNR range, indicating strong robustness to sensing uncertainty. In contrast, the model trained without CRB-guided perturbation is highly sensitive to SNR variations, especially in the low-SNR regime. At 0~dB, the proposed method reduces the MPJPE from about 17.6~cm to 4.6~cm, corresponding to an improvement of approximately 73.9\%. Even at 30~dB, a noticeable gain is still observed, with the MPJPE reduced from about 5.8~cm to 4.0~cm, i.e., an improvement of about 31.0\%. This confirms that the proposed perturbation strategy not only improves low-SNR scenarios, but also provides stable performance gains over the entire sensing quality range. Fig.~\ref{fig:two_subfiguresabout}(\subref{joints}) further visualizes the predicted human pose under different SNR levels, demonstrating the robustness of the proposed CRB-guided strategy under varying sensing conditions.

\subsection{Performance of the ISCC System}

We compare the proposed adaptive scheme with three baselines, labeled \textbf{fixed} $\mathbf{C=1}$, \textbf{fixed $\boldsymbol{\mathrm{SNR}_{s,o}=0}$~dB}, and
\textbf{fixed $\boldsymbol{f_{\mathrm{in}}=0.4}$~GHz}. These baselines isolate the effects of model inference depth, sensing quality, and computation frequency allocation, respectively. The fixed $C=1$ baseline uses the minimum inference depth and prioritizes sensing quality, while the fixed $\mathrm{SNR}_{s,o}=0$~dB baseline only meets the minimum sensing QoS requirement and prioritizes computation resources. The fixed $f_{\mathrm{in}}=0.4$~GHz baseline keeps the computation frequency constant, where the computation cost is adjusted only through the selected model inference depth. 

As shown in Fig.~\ref{fig:optEthree}(\subref{total energy}), the MPJPE generally decreases as the average power budget increases, since additional power provides more resources for improving sensing quality and supporting a more effective inference configuration. The proposed scheme consistently achieves the lowest MPJPE. {In the low-power region, it reduces the MPJPE to about $5.4$~cm, compared with around $6.8$~cm for the fixed $C=1$ baseline, corresponding to an error reduction of approximately $21\%$.
Among the baselines, the fixed $f_{\mathrm{in}}=0.4$~GHz baseline is infeasible at $P_\mathrm{max}=1$~W because its high fixed inference frequency incurs excessive computation energy, while its MPJPE decreases rapidly as the power budget increases because the relaxed energy constraint allows the fixed-frequency configuration to support deeper inference.}
The fixed $C=1$ scheme changes only slightly with the power budget because the minimum model inference depth limits the inference capability. The fixed $\mathrm{SNR}_{s,o}=0$~dB scheme achieves lower MPJPE than the fixed $C=1$ scheme, but its improvement is limited by the minimum sensing-quality constraint.

As shown in Fig.~\ref{fig:optEthree}(\subref{Tsensing}), the MPJPE increases with the sensing duration $T_s$. This is because a larger $T_s$ decreases inference time while increasing inference frequency. 
The proposed adaptive scheme consistently achieves the lowest MPJPE over the whole range of $T_s$. When $T_s$ increases from $20$ ms to $60$ ms, the proposed scheme keeps the MPJPE around $4.1$--$4.3$~cm, while the fixed $C=1$ baseline remains above $6.4$~cm, corresponding to an error reduction of about $33\%$. Even at $T_s=80$~ms, the proposed scheme still achieves a lower MPJPE than the other baselines.  {The fixed $f_\mathrm{in}=0.4$~GHz baseline remains flat at small $T_s$ because sufficient inference time is still available to maintain a similar model depth, but rises sharply once the remaining inference time becomes insufficient under the fixed inference frequency.} This verifies that the proposed scheme can effectively allocate resources to maintain robust pose prediction performance.

{As shown in Fig.~\ref{fig:optEthree}(\subref{Rcommunciation}), the MPJPE increases with the communication rate $R_{\min}$, especially in the high-rate region. A larger $R_{\min}$ forces the system to allocate more resources to communication, reducing the resources available for sensing and computation. The proposed scheme consistently achieves the lowest MPJPE. It maintains an MPJPE of about $4.3$~cm when $R_{\min}$ is below $2.5$~Gbps, reducing the error by approximately $35\%$ compared with the fixed $C=1$ baseline and about $9\%$ compared with the fixed $f_{\mathrm{in}}=0.4$~GHz baseline. Even when $R_{\min}$ increases to $3.5$~Gbps, the proposed scheme still outperforms all baselines. This verifies that adaptive joint allocation can better balance communication QoS and sensing accuracy under stringent rate requirements.}
\begin{figure}[!t]
    \centering
    \includegraphics[width=0.4\textwidth]{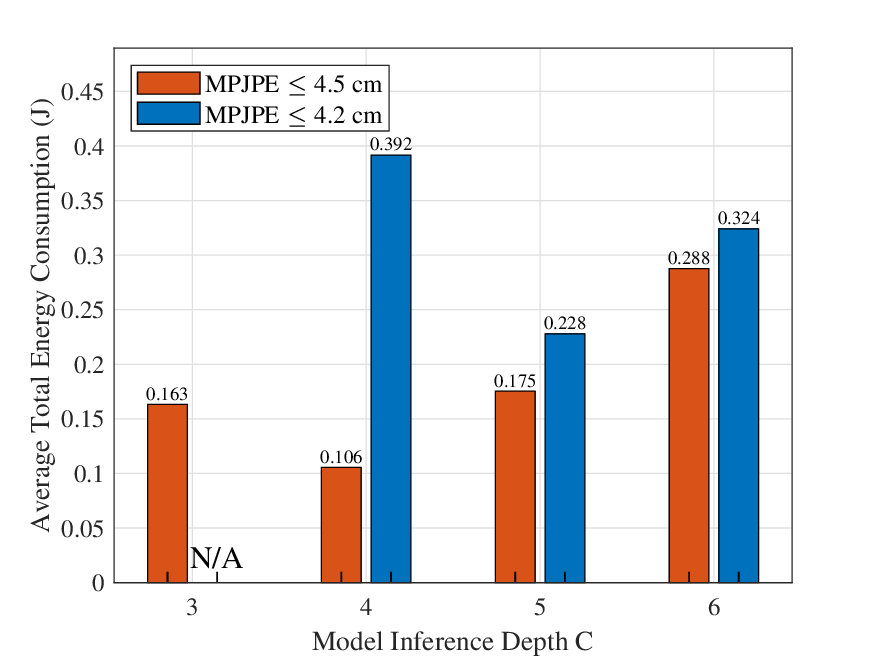}  
    \caption{Average total energy consumption under different MPJPE requirements and model inference depths.}
    \label{fig:depth}
    \vspace{0em}
\end{figure}

{To investigate the energy consumption required to achieve different prediction performance, Fig.~\ref{fig:depth} compares the average total energy consumption under two MPJPE targets, i.e., $4.5$~cm and $4.2$~cm, under different values of $C$. N/A means that the corresponding model inference depth cannot satisfy the MPJPE target of $4.2$~cm, indicating the limited prediction capability of shallow models.
The results show that a stricter performance requirement generally leads to higher energy consumption. For example, when $C=5$, reducing the MPJPE target from $4.5$~cm to $4.2$~cm increases the energy consumption from $0.175$~J to $0.228$~J. However, the energy consumption does not vary monotonically with the model inference depth under the same MPJPE requirement. For instance, under the MPJPE target of $4.5$~cm, increasing the model inference depth from $C=3$ to $C=4$ decreases the energy consumption from $0.163$~J to $0.106$~J, whereas further increasing $C$ leads to higher energy consumption. This is because a shallow model requires more sensing resources to compensate for its limited inference capability, while an excessively deep model introduces additional computation energy consumption. Therefore, these results demonstrate the importance of selecting an appropriate model inference depth to balance sensing and computation energy under different prediction accuracy requirements.}
\begin{figure}[!t]
    \centering    \includegraphics[width=0.75\linewidth]{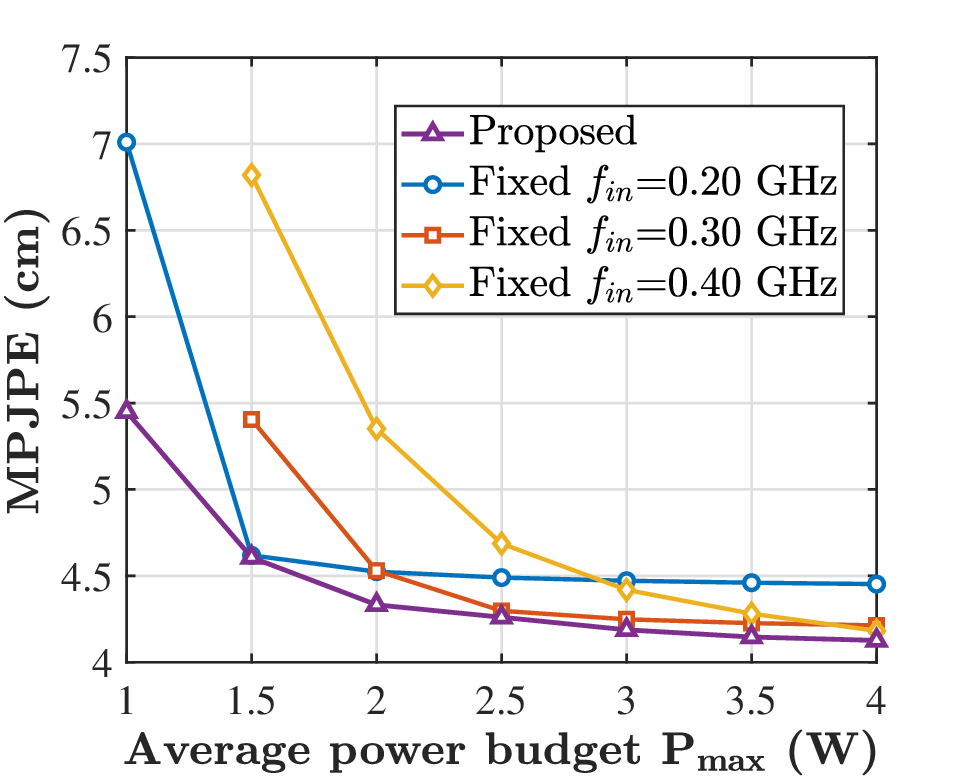}
    \caption{MPJPE performance versus $P_{\mathrm{max}}$ under different inference frequencies.}
    \label{fig:frequency}
    \vspace{-1.5em}
\end{figure}

{To investigate the impact of inference frequency on prediction performance, Fig.~\ref{fig:frequency} compares the proposed scheme with several fixed-frequency baselines.
When $P_{\max}=1.5$~W, the proposed scheme reduces the MPJPE by about $14\%$ and $32\%$ compared with the fixed $f_{\mathrm{in}}=0.30$~GHz and $f_{\mathrm{in}}=0.40$~GHz baselines, respectively. When $P_{\max}=1$~W, these two baselines become infeasible because their relatively high fixed inference frequencies incur excessive computation energy under the stringent power budget. For the fixed $f_{\mathrm{in}}=0.20$~GHz baseline, the MPJPE decreases rapidly as $P_{\max}$ increases from $1$ to $1.5$~W, since the additional energy allows a deeper model to be adopted. When $P_{\max}\geq 2$~W, the performance gradually saturates because the low fixed inference frequency becomes limited by the inference latency.}

{To investigate the impact of model inference depth on prediction performance, Fig.~\ref{fig:depthC} compares the proposed scheme with several fixed-depth baselines. It can be observed that a larger model inference depth generally leads to lower MPJPE, since deeper inference can extract more effective temporal and spatial features for pose prediction. For example, when $P_{\max}=2$~W, the proposed method reduces the MPJPE by about $33\%$, $15\%$, and $4\%$ compared with the fixed $C=1$, $C=2$, and $C=3$ baselines, respectively. When $P_{\max}=1$~W, the fixed $C=3$ baseline becomes infeasible due to its higher computation requirement.}

\begin{figure}[!t]
    \centering
    \includegraphics[width=0.75\linewidth]{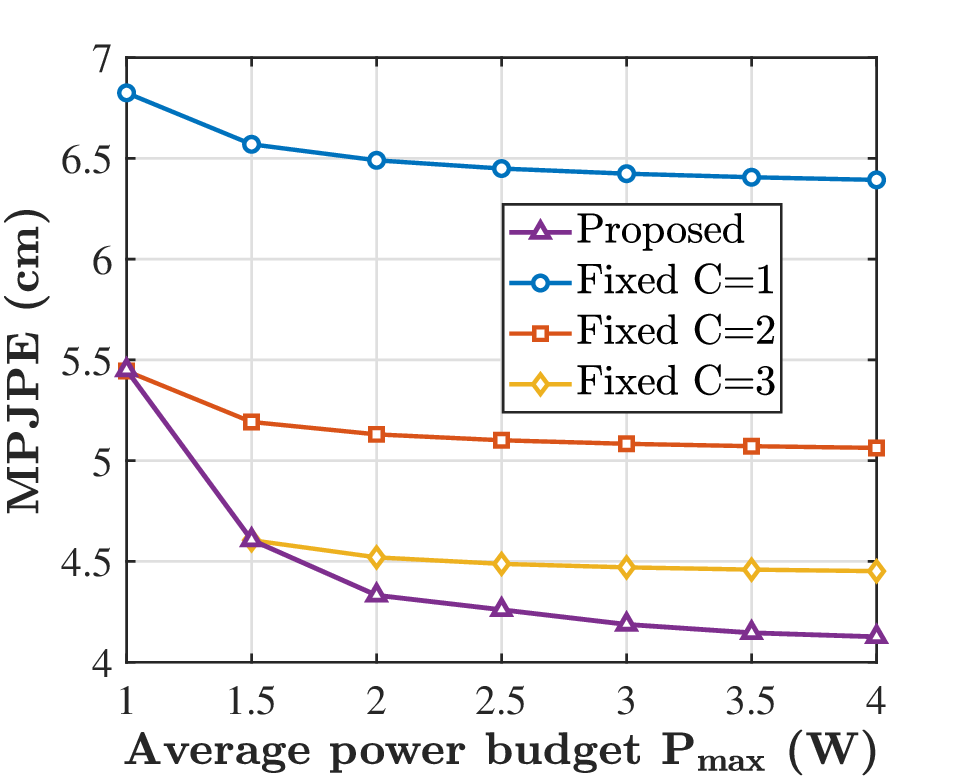}
    \caption{MPJPE performance versus $P_{\mathrm{max}}$ under different model inference depths.}
    \label{fig:depthC}
    \vspace{-1.5em}
\end{figure}

\section{Conclusion}

In this paper, we have proposed a CRB-guided sensing framework for indoor ISCC systems with human pose prediction. We have developed ET-Mamba as a lightweight prediction model that supports adaptive-depth inference under resource-constrained conditions. We have also designed a CRB-guided anisotropic perturbation strategy to simulate point-cloud perturbations under different sensing SNR levels. We have formulated a joint optimization of the beamforming matrix, model inference depth, and computation frequency to minimize the pose prediction error, which is the first study to jointly consider sensing quality, model inference depth and resource allocation in mmWave ISCC systems. To address this issue, we have proposed an alternating optimization
scheme. First, we have established an empirical relationship among pose prediction error, sensing SNR, and model inference depth. Subsequently, we have proposed an AO-based algorithm, where closed-form updates and SDP were combined to efficiently solve the problem. Simulation results show that the proposed method reduces the MPJPE by up to approximately 35\% compared with the baseline, verifying the effectiveness of conducting joint sensing, communication, and computation design in ISCC systems. Future work will extend the proposed framework to multi-person scenarios and online adaptation with real-time sensing feedback.
\bibliographystyle{ieeetr}
\bibliography{bib}

\end{document}